\documentclass[12pt,a4paper]{article}
\usepackage[truedimen,margin=30mm]{geometry} 
\usepackage{mathrsfs}
\usepackage{amssymb}
\usepackage{amsmath}
\usepackage{ascmac}
\usepackage{amsthm}
\usepackage[dvips]{graphicx}
\usepackage{natbib}
\usepackage{setspace}
\usepackage{times}
\usepackage{multirow}
\usepackage{url}
\usepackage{booktabs}
\usepackage{color}
\usepackage{comment}
\usepackage{graphbox}
\usepackage[hidelinks]{hyperref}
\usepackage{xurl}

\usepackage{titlesec}
\titleformat*{\section}{\large\bfseries}
\titleformat*{\subsection}{\it}

\newtheorem{thm}{Theorem}
\newtheorem{lem}{Lemma}
\newtheorem{cor}{Corollary}

\newtheorem{algo}{Algorithm}
\newtheorem{as}{Assumption}

\def\pt{{\widetilde{p}}}
\def\mid{\kern1.5pt|\kern1.5pt}

\title{{\bf Contrastive Bayesian Inference \\for Unnormalized Models}\footnote{\today}}

\date{}

\begin{document}

\maketitle
\doublespacing

\vspace{-1.7cm}
\begin{center}
{\large Naruki Sonobe$^{1, 5}$, Shonosuke Sugasawa$^{2, 5}$, Daichi Mochihashi$^3$ and Takeru Matsuda$^{4,5}$}

\medskip

\medskip
\noindent
$^1$Graduate School of Economics, Keio University\\
$^2$Faculty of Economics, Keio University\\
$^3$The Institute of Statistical Mathematics\\
$^4$Graduate School of Information Science and Technology, The University of Tokyo\\
$^5$RIKEN Center for Brain Science\\
\end{center}

\vspace{0.3cm}
\begin{center}
{\bf \large Abstract}
\end{center} 
Unnormalized (or energy-based) models provide a flexible framework for capturing the characteristics of data with complex dependency structures. However, the application of standard Bayesian inference methods has been severely limited because the parameter-dependent normalizing constant is either analytically intractable or computationally prohibitive to evaluate. 
A promising approach is score-based generalized Bayesian inference, which avoids evaluating the normalizing constant by replacing the likelihood with a scoring rule. However, this approach still requires careful tuning of the likelihood information, and it may fail to yield valid inference without appropriate control.
To overcome this difficulty, we propose a fully Bayesian framework for inference on unnormalized models that does not require such tuning. 
We build on noise-contrastive estimation, which recasts inference as a binary classification problem between observed and noise samples, and treat the normalizing constant as an additional unknown parameter within the resulting likelihood. For exponential families, the classification likelihood becomes conditionally Gaussian via Pólya–Gamma data augmentation, leading to a simple Gibbs sampler. 
We further establish posterior concentration and a Bernstein--von Mises theorem for our proposed method, providing theoretical justification for its uncertainty quantification.
We demonstrate the proposed approach through two models: time-varying density models of temporal point processes and sparse torus graph models of multivariate circular data. 
Through simulation studies and real-data analyses, our proposed method provides accurate point estimates and principled uncertainty quantification.

\vspace{-0cm}


\section{Introduction}
\label{sec:introduction}

Bayesian inference for unnormalized models, \textit{i.e.}, statistical models whose likelihoods contain intractable normalizing constants, remains a persistent challenge. The core obstacle is evaluating the parameter-dependent normalizing constant, which usually requires integrating over the whole sample space and thus is often analytically intractable. 
Unnormalized models arise in a range of applications: for instance, Ising models \citep{besag1974spatial}, exponential random graph models \citep{robins2007network, hunter2012network}, non-Gaussian graphical models \citep{lin2016estimation,inouye2017WIREs}, nonlinear independent component analysis \citep{hyvarinen2024identifiability,matsuda2019estimation}, torus graph models \citep{klein2020torus}, and time-varying density estimation based on logistic Gaussian processes \citep{lenk1988logistic, tokdar2007posterior}. Despite this wide range of applications, the intractable likelihood implies that both the likelihood and the resulting posterior distribution are only known up to parameter-dependent normalizing constants that cannot be evaluated, thereby rendering standard Bayesian inference methods inapplicable.

A wide range of computational methods has been proposed to perform Bayesian inference for unnormalized models with intractable normalizing constants. \cite{park2018intractable} provided a comprehensive review of Markov chain Monte Carlo (MCMC) methods for such models. Representative exact approaches include auxiliary-variable MCMC
methods such as the exchange algorithm
\citep{moller2006markov,murray2006mcmc}, and pseudo-marginal
methods based on unbiased likelihood estimators
\citep{andrieu2009pseudo,lyne2015intractable}.
Although theoretically exact under their respective conditions, these
methods can be computationally demanding because they require repeated
model simulation or Monte Carlo estimation within MCMC iterations
\citep{park2018intractable}.
To alleviate this computational burden, several asymptotically inexact or approximate MCMC algorithms have also been developed. These methods, such as the function emulation approach proposed by \cite{park2020function}, can offer substantial gains in computational efficiency by introducing controlled approximations to the likelihood or normalizing constant. While these approaches are often much faster, they come with a critical trade-off: the stationary distribution of the resulting Markov chain typically differs from the exact target posterior, and there is generally no strong theoretical guarantee that the generated samples will converge to the true posterior distribution asymptotically.
In parallel to these MCMC-based approaches, likelihood-free inference methods exploit situations in which it is much easier to generate synthetic data from the model than to evaluate its (unnormalized) likelihood. A canonical example is approximate Bayesian computation (ABC), which approximates the posterior by comparing summary statistics computed from the observed data with those from simulated data \citep{marin2012ABC}. Further connections between the proposed framework and likelihood-free inference methods are discussed in Section~\ref{sec:discussion}.

Additionally, there are recent developments in generalized Bayesian approaches that update prior beliefs using task-specific losses or proper scoring rules in place of a standard likelihood \citep{bissiri2016general}. 
An appealing feature of this loss-based formulation is that certain choices avoid evaluating normalizing constants, making the framework suitable for unnormalized models.
Illustrative examples include methods based on the Stein discrepancy \citep{matsubara2022robust}, Hyv\"arinen score (also known as the Fisher divergence) \citep{jewson2022general, altamirano2023score, altamirano2024robust, matsubara2024generalized} and the log-ratio matching divergence \citep{laplante2025general}, which support Bayesian-style estimation without likelihood normalization. 
Relatedly, \cite{pacchiardi2024generalized} developed a likelihood-free framework based on scoring rule posteriors, connecting scoring rule-based inference with simulator-based models and demonstrating how several existing likelihood-free methods can be recovered as special cases.
These approaches are often computationally efficient and come with solid theoretical underpinnings; meanwhile, they introduce a learning-rate or tuning hyperparameter that meaningfully affects the resulting generalized posterior and, thus, inference and uncertainty quantification, motivating principled, data-driven calibration \citep{bissiri2016general,syring2019general,lyddon2019general, wu2023Generalized}.

In this paper, we develop a fully Bayesian formulation of noise-contrastive estimation (NCE) for inference in unnormalized models, sidestepping the direct computation of the normalizing constant. Our approach leverages the principles of NCE \citep{gutmann2010noise,gutmann2012noise}, which reframes the likelihood construction as a binary classification problem between observed data and artificially generated noise. 
While NCE has demonstrated considerable success across a range of applications, most existing methods estimate parameters by directly maximizing the NCE objective. A few works introduce Bayesian or variational aspects, such as variational NCE \citep{rhodes2019variational} and fully variational NCE \citep{zach2023fully}, but these still represent parameter uncertainty through variational approximations rather than a fully Bayesian posterior over model parameters. 
To illustrate the practical benefits of our fully Bayesian NCE framework for unnormalized models, we highlight two key advantages over existing approaches. (1) Full Bayesian inference provides principled uncertainty quantification over all model parameters and enables the natural incorporation of latent variables, allowing for flexible and interpretable probabilistic modeling. (2) While score-based generalized Bayesian inference has shown promise for unnormalized models, its reliance on score matching makes information control difficult, rendering it ill-suited for data with hierarchical structures or shrinkage priors; our framework overcomes this limitation by operating directly within a proper Bayesian posterior framework. We demonstrate these advantages in two challenging settings: time-varying density estimation, where uncertainty over latent temporal dynamics is critical, and edge selection for sparse graphical models of multivariate angular data, where structured sparsity must be carefully encoded through the prior.

This paper is organized as follows. Section \ref{sec:methodology} details our proposed methodology, which combines NCE with P\'olya--Gamma data augmentation for efficient Bayesian inference, and proposes an appropriate selection of noise distributions. Section \ref{sec:theoretical-assessment} establishes the asymptotic properties of NC-Bayes: posterior concentration and Bernstein--von Mises results. In Section \ref{sec:nc-bayes-action1}, we apply this framework to the problem of time-varying density estimation. Section \ref{sec:nc-bayes-action2} demonstrates the applicability of the framework through a second application: learning sparse graph structure for multivariate circular data. Finally, Section \ref{sec:discussion} concludes the paper with a summary of our findings and a discussion of future work.

\section{Methodology}
\label{sec:methodology}
\subsection{Unnormalized models and classification likelihood}\label{sec:unnormalized-classification}

Consider the problem of estimating an unknown parameter (vector) $\theta\in\Theta\subset\mathbb{R}^p$ in a parametric model $p(x\mid\theta)$ from observations $x_1,\ldots,x_n$. We focus on models whose likelihood involves an intractable normalizing constant. Such a model can be expressed as
$$
p(x\mid\theta)=\frac{\tilde{p}(x\mid\theta)}{Z(\theta)}, \quad
Z(\theta)=\int \tilde{p}(x\mid\theta) dx,
$$ 
where the parameter $\theta$ controls the shape or features of the model, $\tilde{p}(x\mid\theta)$ is an unnormalized probabilistic model, and $Z(\theta)$ is the normalizing constant, or partition function, that converts $\tilde{p}(x\mid\theta)$ into a proper probability distribution. In many applications, $x$ is high-dimensional and the model induces strong dependence or interactions among components, thus $Z(\theta)$ does not factorize and has no closed-form expression. Moreover, accurate numerical evaluation is prohibitively expensive. As a result, the likelihood $p(x\mid\theta)$ cannot be computed pointwise for arbitrary $\theta$, and standard Bayesian procedures that require repeated evaluation of the likelihood, or ratios of normalizing constants across parameter values, cannot be directly applied.

In the framework of NCE \citep{gutmann2010noise, gutmann2012noise}, density estimation under such a model is reformulated as a binary classification that distinguishes between observed data and artificially generated noise $x_{n+1},\ldots,x_{n+m}$ sampled from a known distribution $q(x)$, which we refer to as the noise distribution. Given $x$, the probability that $x$ is recognized as a genuine observation is given by 
$$
r(x\mid\theta, Z(\theta))=\frac{n \tilde{p}(x\mid\theta)/Z(\theta)}{n \tilde{p}(x\mid\theta)/Z(\theta)+mq(x)}
=\frac{n \tilde{p}(x\mid\theta)}{n \tilde{p}(x\mid\theta)+mZ(\theta)q(x)}\,.
$$
Then, the likelihood for the binary classification between the genuine and artificial observations is obtained as 
\begin{equation}\label{eq:bin-likelihood}
\begin{split}
L(\theta, Z(\theta)\mid X_n, X_m^\ast) 
&\equiv
\prod_{i=1}^n r(x_i\mid\theta, Z(\theta))\prod_{i=n+1}^{n+m} \big\{1-r(x_i\mid\theta, Z(\theta))\big\}\\
&=\prod_{i=1}^{n+m}\frac{\{n \tilde{p}(x_i\mid\theta)\}^{s_i}\{mZ(\theta)q(x_i)\}^{1-s_i}  }{n \tilde{p}(x_i\mid\theta)+mZ(\theta)q(x_i)},
\end{split}
\end{equation}
where $X_n=\{x_1,\ldots,x_n\}$ is a set of genuine observations, $X_m^\ast=\{x_{n+1},\ldots,x_{n+m}\}$ is a set of noise observations, and $s_i=\mathbb{I}(i\leq n)$ is an indicator of a genuine observation. 
This likelihood corresponds to that of a logistic regression classifier, in which the parameter estimates are obtained by maximizing the ability to discriminate between genuine and artificial observations.

To enable Bayesian inference, we place priors on both the parameter $\theta$ and the normalizing constant $Z(\theta)$. We treat $Z(\theta)$ as a parameter separate from $\theta$, and hereafter denote it by $Z$. This specification of prior distributions and the classification likelihood (\ref{eq:bin-likelihood}) gives the posterior distribution of $(\theta^\top, Z)^\top$ as 
\begin{equation}\label{eq:pos}
\pi(\theta, Z\mid X_n, X_m^\ast)\propto \pi(\theta, Z)L(\theta, Z\mid X_n, X_m^\ast), 
\end{equation}
where $\pi(\theta, Z)$ is a prior distribution for $(\theta^\top, Z)^\top$. 
Note that the posterior distribution above is free from the intractable normalizing constant, owing to the classification likelihood (\ref{eq:bin-likelihood}) and does not include a tuning constant, unlike generalized Bayesian approaches, as further discussed in Section \ref{sec:theoretical-assessment}. We refer to this procedure as noise-contrastive Bayes ({\bf NC-Bayes}).
Hence, the posterior computation of (\ref{eq:pos}) can be performed via general algorithms such as Hamiltonian Monte Carlo methods. 
On the other hand, when the observation model $p(x\mid\theta)$ belongs to an exponential family, efficient sampling algorithms can be developed by employing P\'olya--Gamma data augmentation \citep{polson2013bayesian}, as described in the subsequent section. 

\subsection{Gibbs sampler for exponential families}\label{sec:Gibbs}
Here, we describe the details of posterior computation for the binary likelihood induced by NCE for an exponential-family observation model,
$
\tilde{p}(x\mid\theta)=h(x)\exp\{\eta(x)^\top\theta\},
$
where $h(x)$ is a known base density or base measure and $\eta(x)$ is a vector of sufficient statistics, with $\theta$ being the corresponding natural parameter.
Note that the exponential-family form is highly expressive and covers a broad range of widely used models, including generalized linear models such as linear regression, logistic regression and Poisson regression, as well as Ising models \citep{besag1974spatial}, exponential random graph models \citep{robins2007network, hunter2012network} and truncated Gaussian graphical models \citep{lin2016estimation}. Importantly, the model considered in our experiments also falls within this framework. Therefore, the proposed posterior computation applies generically to a wide range of regression and graphical models whenever the (unnormalized) likelihood admits an exponential-family representation.

For the exponential-family model, the binary likelihood is 
\begin{align}
&\prod_{i=1}^{n+m}\frac{\exp(\log n + \log h(x_i) + \eta(x_i)^\top \theta)^{s_i}\{mZq(x_i)\}^{1-s_i}  }{\exp(\log n + \log h(x_i) + \eta(x_i)^\top \theta)+mZq(x_i)}
=\prod_{i=1}^{n+m} \frac{(e^{\psi_i})^{s_i}}{1+e^{\psi_i}}\,,
\label{eqn:binary-likelihood}
\end{align}
where 
$\psi_i=\log n + \log h(x_i) +\eta(x_i)^\top\theta -\log m -\log Z -\log q(x_i)$.
Let us introduce the reparametrization $\beta=-\log Z$, $z(x_i)=(\eta(x_i)^\top, 1)^\top$, $\gamma=(\theta^\top, \beta)^\top$ and $C(x_i)=\log n -\log m + \log h(x_i) -\log q(x_i)$ so that $\psi_i=z(x_i)^\top\gamma + C(x_i)$.
Using the integral expression in \cite{polson2013bayesian}, the likelihood \eqref{eqn:binary-likelihood}
can be augmented as 
\begin{align*}
\prod_{i=1}^{n+m} \frac{\exp(z(x_i)^\top\gamma + C(x_i))^{s_i}}{1+\exp(z(x_i)^\top\gamma + C(x_i))}
&=\prod_{i=1}^{n+m}\exp\left\{\Big(s_i-\frac{1}{2}\Big) (z(x_i)^\top\gamma + C(x_i)) \right\}
\\
&\times \frac12\int_0^{\infty} \exp\left\{-\frac12\omega_i(z(x_i)^\top\gamma + C(x_i))^2\right\}p_{\rm PG}(\omega_i)d\omega_i\,,
\end{align*}
where $p_{\rm PG}(\omega_i)$ is the density of a P\'olya--Gamma distribution, ${\rm PG}(1,0)$. 
Since the likelihood is now written as a scale mixture of Gaussians, if we assume a Gaussian prior for $\gamma$, the full conditional posterior of $\gamma$ also becomes Gaussian and the full conditional posterior of $\omega_i$ is a P\'olya--Gamma distribution. Therefore, posterior samples can be obtained via Gibbs sampling. Suppose that the joint prior for $\gamma$ is $N(A_0, B_0)$. 
Then, the posterior samples can be generated by Gibbs sampling as follows: 

\begin{algo}\label{algo:gibbs}
Starting with the initial value of $\gamma$, iteratively generate samples of $\gamma$ and the auxiliary variables $\omega_i$ as follows: 
\begin{enumerate}
\item
For $i=1,\ldots,n+m$, generate $\omega_i$ from ${\rm PG}(1, z(x_i)^\top\gamma + C(x_i))$.

\item
Generate $\gamma$ from $N(A_1, B_1)$, where 
\begin{equation*}
\left\{\kern-1.5ex
\begin{array}{ll}
&\displaystyle B_1=\left(B_0^{-1} + \sum_{i=1}^{n+m}\omega_i z(x_i)z(x_i)^\top\right)^{-1}, \\
&\displaystyle A_1=B_1\left\{\sum_{i=1}^{n+m}\Big(s_i-\frac{1}{2}-\omega_iC(x_i)\Big) z(x_i) + B_0^{-1}A_0\right\}.  
\end{array}
\right.
\end{equation*}
\end{enumerate}
\end{algo}

\subsection{Marginalization and adaptive updating of the noise distribution}\label{marginalization-updating}
NC-Bayes depends on the choice of the noise distribution $q(x)$ in terms of statistical efficiency and the stability of estimation. A common practical strategy is to choose $q(x)$ to be close to the data-generating distribution so that the resulting classification task is well balanced and yields more informative comparisons. Recent analyses also indicate that the theoretically optimal noise need not coincide with the data distribution \citep{chehab2022noise}, but we do not pursue this direction further here. 
We begin by choosing $q(x)$ to be a fixed noise distribution, such as a uniform distribution, since the asymptotic theory in Section~\ref{sec:theoretical-assessment} does not require a restrictive specification of the noise distribution.

The standard algorithm for NC-Bayes uses only a single set of noise data, that is, the noise sample is generated in advance and fixed in the analysis. 
In the Bayesian framework, it is also possible to integrate the noise data out with respect to the generative distribution $q(\cdot)$, which makes the posterior inference less sensitive to the noise data. 
Note that $X_n$ and $X_m^\ast$ are vectors of observed data and noise data, respectively.   
Algorithm~\ref{algo:gibbs} targets the posterior distribution given the noise data:
$$
\pi(\gamma\mid X_n,X_m^\ast) 
\propto \pi(\gamma) \prod_{i=1}^{n} \frac{\exp(z(x_i)^\top\gamma + C(x_i))}{1+\exp(z(x_i)^\top\gamma + C(x_i))}
\prod_{i=n+1}^{n+m}\frac{1}{1+\exp(z(x_i)^\top\gamma + C(x_i))}.
$$
On the other hand, to make posterior inference less sensitive to the choice of $X_m^\ast$, we can target the posterior of the form: 
\begin{equation}\label{eq:pos-noise}
\pi(\gamma\mid X_n)=\int \pi(\gamma\mid X_n,X_m^\ast) \prod_{i=n+1}^{n+m}q(x_i)dx_i.
\end{equation}
By generating fresh noise samples at each iteration, this formulation reduces the dependence of posterior inference on any particular realization of noise data.
To generate the integrated posterior (\ref{eq:pos-noise}), we use the following sampling algorithm: 

\begin{algo}\label{algo:gibbs2}
Starting with the initial value of $\gamma$, iteratively generate random samples of $\gamma$ and $\omega_i$ (auxiliary latent variable) as follows: 
\begin{enumerate}
\item
For $i=n+1,\ldots,n+m$, generate noise data $x_i$ from the generative distribution $q(\cdot)$, and compute $z(x_i)$ and $C(x_i)$.

\item
Generate ~$\omega_i$~ from ${\rm PG}(1, z(x_i)^\top\gamma + C(x_i))$ for $i=1,\ldots,n$, and from \\
${\rm PG}(1, z(x_i)^\top\gamma + C(x_i))$ for $i=n+1,\ldots,n+m$.

\item
Generate $\gamma$ from $N(A_1, B_1)$, where 
\begin{equation*}
\left\{\kern-1.5ex
\begin{array}{ll}
&\displaystyle B_1=\left(B_0^{-1} + \sum_{i=1}^{n}\omega_i z(x_i)z(x_i)^\top + \sum_{i=n+1}^{n+m}\omega_i z(x_i)z(x_i)^\top\right)^{-1}, \\
&\displaystyle A_1=B_1\left\{\sum_{i=1}^{n}\Big(\frac{1}{2}-\omega_iC(x_i)\Big) z(x_i) + \sum_{i=n+1}^{n+m}\Big(-\frac{1}{2}-\omega_iC(x_i)\Big) z(x_i) + B_0^{-1}A_0\right\}.  
\end{array}
\right.
\end{equation*}
\end{enumerate}
\end{algo}

At each iteration, we redraw a fresh set of noise samples from $q(\cdot)$ to reduce sensitivity to any particular realization of the noise sample. This is a practical device rather than a requirement for the asymptotic theory in Section~\ref{sec:theoretical-assessment}.

Further, we propose an adaptive method for updating the noise distribution by updating it within MCMC iterations. 
Let $\widetilde{\gamma}$ be the average of the posterior samples of $\gamma$ over a batch of MCMC iterations.
Then, we can specify the noise distribution as 
$$
q_\alpha(x)=\frac{\exp\{\alpha z(x)^\top \widetilde{\gamma}\}}{Z_\alpha}, \qquad 
Z_\alpha=\int \exp\{\alpha z(u)^\top \widetilde{\gamma}\} du,
$$
where $\alpha\in[0,1]$ is a shrinkage (tempering) parameter that controls the concentration of the noise distribution. 
Note that $\alpha=1$ leads to the same distribution as the fitted model while $\alpha=0$ reduces to the uniform distribution.
Samples from $q_\alpha$ can be readily generated by importance resampling.  
To generate $m$ samples from $q_\alpha(x)$, we first generate $M(\gg m)$ samples from a base distribution (e.g. uniform distribution), denoted by $q_0(x)$, and re-sample $m$ units with probabilities proportional to $q_\alpha(x)/q_0(x)$. 
The resulting adaptive noise update is summarized as follows:

\begin{algo}\label{algo:adaptive_noise}
\textbf{(Adaptive noise update by tempered importance resampling)}
\begin{enumerate}
\item
Generate proposals $\{x^{(j)}\}_{j=1}^M$ from $q_0(\cdot)$.

\item
For each $j=1,\ldots,M$, compute $z(x^{(j)})$, set $\widetilde{w}_j=\exp\{\alpha z(x^{(j)})^\top \widetilde{\gamma}\}/q_0(x^{(j)})$ and compute $\widehat{Z}_\alpha=M^{-1}\sum_{j=1}^M \widetilde{w}_j$.

\item  
Resample $m$ points from $\{x^{(j)}\}_{j=1}^M$ with probabilities $\widetilde w_j/\sum_{l=1}^M \widetilde w_l$ $(j=1, \dots, M)$,
and denote the resulting set by $X_m^\ast=\{x_{n+1},\ldots,x_{n+m}\}$.

\item
Compute $q_\alpha(x^{(j)})=\exp\{\alpha z(x^{(j)})^\top \widetilde{\gamma}\}/\widehat{Z}_\alpha$.
\end{enumerate}
\end{algo}

Although the target noise density $q_\alpha(x)$ involves the normalizing constant
$Z_\alpha$, Algorithm~\ref{algo:adaptive_noise} does not require repeated evaluations of $Z_\alpha$ for each resampled point.
The normalization is implicitly handled through the importance resampling step, since the probabilities $\widetilde w_j/\sum_{l=1}^M \widetilde w_l$ depend only on
ratios of unnormalized weights.
The Monte Carlo estimate $\widehat Z_\alpha$ is therefore computed once per noise update and reused for evaluating $q_\alpha$ if needed.
To stabilize the performance of the algorithm above, it is useful to monitor the
effective sample size $\mathrm{ESS}=\left(\sum_{j=1}^M \widetilde w_j\right)^2 / \sum_{j=1}^M \widetilde w_j^2$ during the noise update.
The base distribution $q_0$ can be chosen flexibly, provided that it has support covering the region where $q_\alpha$ has non-negligible mass.
In our implementation, we use a uniform distribution over a data-adaptive bounding box for simplicity.

\subsection{Contrastive Bayesian inference for hierarchical models}\label{sec:hierarchical}

We also consider a multi-group setting with $J$ groups. For $j=1,\ldots,J$, let $X_j=\{x_{j1},\ldots,x_{jn_j}\}$ denote observed data and $X_j^\ast=\{x_{j,n_j+1},\ldots,x_{j,n_j+m_j}\}$ denote artificial noise generated from a possibly group-specific distribution $q_j(\cdot)$. We assume an unnormalized model $p_j(x\mid\theta_j)= \tilde{p}_j(x\mid\theta_j) / Z_j$, where $Z_j=\int \tilde{p}_j(x\mid\theta_j)dx$. We focus on the exponential-family case where $\tilde p_j(x\mid\theta_j)=h(x)\exp(\eta(x)^\top \theta_j)$.
By reparametrizing $\beta_j = -\log Z_j$ and defining vectors $z(x_{ji})=\big(\eta(x_{ji})^\top,\,1\big)^\top$ and $\gamma_j=(\theta_j^\top,\beta_j)^\top$, the classification likelihood for all groups can be expressed as
$$
\prod_{j=1}^J \prod_{i=1}^{n_j+m_j}
\frac{\exp(z(x_{ji})^\top\gamma_j+C(x_{ji}))^{s_{ji}}}{1+\exp(z(x_{ji})^\top\gamma_j+C(x_{ji}))},
$$
where $C(x_{ji})=\log n_j-\log m_j+\log h(x_{ji})-\log q_j(x_{ji})$ and $s_{ji}=I(i\le n_j)$.

To facilitate sharing statistical strength across groups, we introduce a hierarchical prior structure. Specifically, we assume that the group-specific parameters $\{\gamma_j\}_{j=1}^J$ are drawn from a common population distribution:
$$
\theta_j\mid \mu,\Sigma \sim N(\mu,\Sigma), ~
\beta_j\mid \mu_\beta,\sigma_\beta^2 \sim N(\mu_\beta,\sigma_\beta^2),
$$
for $j=1,\ldots,J$. We complete the model by placing standard noninformative or weakly informative hyperpriors on the population parameters $(\mu, \Sigma, \mu_\beta, \sigma_\beta^2)$.

Posterior computation can be performed efficiently via a Gibbs sampler that leverages P\'olya--Gamma data augmentation, which renders the conditional posteriors for the parameters Gaussian. The complete sampling algorithm proceeds as follows:

\begin{algo}\label{algo:gibbs-hier}
Starting from initial values of $\{\gamma_j\}_{j=1}^J$ and the hyperparameters, iteratively generate posterior samples as follows:
\begin{enumerate}
\item For each group $j$, either fix the noise data $X_j^\ast$ throughout all iterations or regenerate $x_{ji}\sim q_j(\cdot)$ for $i=n_j+1,\ldots,n_j+m_j$ at each iteration, recomputing $z(x_{ji})$ and $C(x_{ji})$ accordingly.

\item For all $j=1,\ldots,J$ and $i=1,\ldots,n_j+m_j$, sample the P\'olya--Gamma latent variables:
$$
\omega_{ji}\sim \mathrm{PG}\!\left(1,\,z(x_{ji})^\top\gamma_j+C(x_{ji})\right).
$$

\item For each group $j=1,\ldots,J$, sample the group-specific parameters $\gamma_j = (\theta_j^\top, \beta_j)^\top$ from its Gaussian conditional posterior $N(A_{1j}, B_{1j})$, where
\begin{equation*}
\left\{\kern-1.5ex
\begin{array}{ll}
&\displaystyle B_{1j}=\Big(\mathrm{diag}(\Sigma^{-1}, \sigma_\beta^{-2})+\sum_{i=1}^{n_j+m_j}\omega_{ji}\,z(x_{ji})z(x_{ji})^\top\Big)^{-1},\\
&\displaystyle A_{1j}=B_{1j}\Big\{\sum_{i=1}^{n_j+m_j}\!\Big(s_{ji}-\frac12-\omega_{ji}C(x_{ji})\Big)z(x_{ji})
+ \mathrm{diag}(\Sigma^{-1}, \sigma_\beta^{-2})(\mu^\top, \mu_\beta)^\top \Big\}.
\end{array}
\right.
\end{equation*}

\item Sample the hyperparameters $(\mu,\Sigma,\mu_\beta,\sigma_\beta^2)$ from their respective full conditional distributions, given the current values of $\{\theta_j, \beta_j\}_{j=1}^J$ and the hyperpriors. Assuming standard conjugate hyperpriors (e.g., Gaussian for means, Inverse-Wishart/Gamma for variances), these updates are standard.
\end{enumerate}
\end{algo}

Under this hierarchical specification, applying NC-Bayes naturally captures between-group heterogeneity via Bayesian partial pooling. The group-specific parameters $\{\theta_j,\beta_j\}$ as well as the hyperparameters $(\mu,\Sigma,\mu_\beta,\sigma_\beta^2)$ can be updated within a unified and efficient Gibbs sampler using the P\'olya--Gamma data augmentation.

\section{Theoretical Assessment}
\label{sec:theoretical-assessment}

\subsection{Notation}
\label{sec:notation}
This section establishes the asymptotic properties of the NC-Bayes
posterior under a fixed noise distribution $q$. We consider the fixed-ratio
asymptotic regime of noise-contrastive estimation, in which the number of
noise observations is proportional to the number of data observations
\citep{gutmann2010noise,gutmann2012noise}. The proof strategy follows the
generalized-posterior framework of \citet{miller2021general} and
\citet{matsubara2022robust}.
The regularity conditions and all proofs are deferred to Supplementary Material~\ref{sec:s1}.

Let $X_n={x_1,\ldots,x_n}$ denote observations drawn independently from the
data-generating distribution $p_0$, and let
$X_m^\ast={x_{n+1},\ldots,x_{n+m}}$ denote an independent sample from the
noise distribution $q$. Fix $\nu\in\mathbb Q_{>0}$ and consider sample sizes
$n$ for which $m=\nu n\in\mathbb N$.
We write the unnormalized model as
$p(x\mid\gamma)=\tilde p(x\mid\theta)\exp(\beta)$, where
$\gamma=(\theta^\top,\beta)^\top\in\Gamma\subset\mathbb R^{p+1}$. For the
fixed-$\nu$ asymptotic theory, $\Gamma$ is the open, bounded, and convex
parameter set specified in Assumption~\ref{as:nc-bayes} in Supplementary
Material~\ref{sec:s1}.
Let $\lambda$ be a dominating $\sigma$-finite measure and define
$\mathcal X_0=\{x:p_0(x)>0\}$. The correctly normalized value satisfies
$\beta=-\log Z(\theta)$, where
$Z(\theta)=\int\tilde p(x\mid\theta)\,dx$.
Assume that $q(x)>0$ for $\lambda$-almost every $x\in\mathcal X_0$. We further assume that, for every $\theta$ under consideration,
$\tilde p(x\mid\theta)>0$ for $\lambda$-almost every $x\in\mathcal X_0$ and
$\tilde p(x\mid\theta)=0$ for $\lambda$-almost every $x\notin\mathcal X_0$.
The support of $q$ may be strictly larger than $\mathcal X_0$.
On $\mathcal X_0$, define the log density ratio and the corresponding
conditional class probability by
$a(x\mid\gamma)=\log\tilde p(x\mid\theta)+\beta-\log q(x)$ and
$\rho_\nu(x\mid\gamma)=[1+\nu\exp\{-a(x\mid\gamma)\}]^{-1}$.

The NC-Bayes posterior measure $\Pi_n$, also written
$\Pi(\cdot\mid X_n,X_m^\ast)$, is
$$
\Pi_n(d\gamma)
=
\Pi(d\gamma\mid X_n,X_m^\ast)
\propto
\pi(\gamma)\exp\{-nf_{n,\nu}(\gamma)\}\,d\gamma,
$$
where
\begin{equation}
\label{eq:main-fn}
f_{n,\nu}(\gamma)
=
-\frac{1}{n}\sum_{i=1}^n
\log\rho_\nu(x_i\mid\gamma)
-
\frac{1}{n}\sum_{i=n+1}^{n+m}
\log\{1-\rho_\nu(x_i\mid\gamma)\}.
\end{equation}
Thus, the NC-Bayes posterior is a generalized posterior based on the empirical
noise-contrastive loss (NC-loss).  
In this section, we give asymptotic properties of the $\theta$-marginal NC-Bayes posterior, and the corresponding result for the joint posterior is proved in Supplementary Material \ref{sec:s1}. Let $R_\theta=(I_p,0)\in\mathbb R^{p\times(p+1)}$ denote the projection from
$\gamma=(\theta^\top,\beta)^\top$ to $\theta$, and let $R_\theta^{-1}$ denote the inverse-image operator associated with $R_\theta$.

\subsection{Posterior concentration}
\label{sec:posterior-concentration}

Our first result establishes posterior concentration.
Under the regularity conditions stated in Supplementary Material \ref{sec:s1},
the population NC-loss admits a well-separated minimizer at the true
parameter.  Consequently, the joint NC-Bayes posterior concentrates around
$\gamma_0$, and hence its $\theta$-marginal concentrates around
$\theta_0$ as the sample size increases.
Assume that the data-generating distribution is correctly specified:
$$
p_0(x)
=
\frac{\tilde p(x\mid\theta_0)}{Z(\theta_0)},
\qquad
\gamma_0
=
(\theta_0^\top,-\log Z(\theta_0))^\top.
$$
Let
\begin{equation}
\label{eq:main-fpop}
f_\nu(\gamma)
=
-\mathbb E_{p_0}\log\rho_\nu(X\mid\gamma)
-
\nu\mathbb E_q
\log\{1-\rho_\nu(X^\ast\mid\gamma)\}
\end{equation}
denote the population NC-loss.  Assumption~\textnormal{(A3)(ii)} directly requires that $\gamma_0$ be a
well-separated minimizer of $f_\nu$ over $\Gamma$.  Together with uniform
convergence of $f_{n,\nu}$ to $f_\nu$, this condition yields posterior
concentration and is also used to control the tails in the
Bernstein--von Mises argument below.

\begin{thm}[Posterior concentration]
\label{thm:conc}
Suppose Assumptions~\textnormal{(A1)}--\textnormal{(A4)} in Supplementary
Material~\ref{sec:s1} hold with $r_{\max}=1$.  Let $\Pi_n^\theta=\Pi_n\circ R_\theta^{-1}$ denote the
$\theta$-marginal NC-Bayes posterior. Then, for every $\epsilon>0$,
$$
\Pi_n^\theta
(\{\theta\in\Theta:\|\theta-\theta_0\|<\epsilon\})
\overset{\text{a.s.}}\to1.
$$
\end{thm}

\subsection{Bernstein--von Mises}
\label{sec:BvM}

Let $\hat\gamma_n$ denote the measurable minimum NC-loss estimator specified
in Assumption~\textnormal{(A3)(i)}, that is,
$\hat\gamma_n\in
\operatorname*{arg\,min}_{\gamma\in\Gamma}f_{n,\nu}(\gamma)$.
Define
$s(x)=\left.\nabla_\gamma a(x\mid\gamma)\right|_{\gamma=\gamma_0}$,
$\rho_{0, \nu}(x)=\rho_\nu(x\mid\gamma_0)$, and
$I_\nu=\mathbb E_{p_0}
[\{1-\rho_{0, \nu}(X)\}s(X)s(X)^\top]$.
Supplementary Lemma~\ref{lem:deriv-ident} shows that
$I_\nu=\nabla_\gamma^2f_\nu(\gamma_0)$, so $I_\nu$ is the local curvature of
the population NC-loss at $\gamma_0$.
We set
$\hat\theta_n=R_\theta\hat\gamma_n$
and $V_{\nu,\theta}=R_\theta I_\nu^{-1}R_\theta^\top$.
To establish the Bernstein--von Mises result, we apply
\citet[Theorem~4]{miller2021general}, following the verification strategy of
\citet[Supplementary Appendix~B.5]{matsubara2022robust}.

\begin{thm}[Bernstein--von Mises]
\label{thm:bvm}
Suppose Assumptions~\textnormal{(A1)}--\textnormal{(A4)} in Supplementary
Material~\ref{sec:s1} hold with $r_{\max}=3$. Then the following statements
hold. For the parameter of interest $\theta$,
$$
\left\|
\Pi\left(
\sqrt n(\theta-\hat\theta_n)\in\cdot
\mid X_n,X_m^\ast
\right)
-
N(0,V_{\nu,\theta})(\cdot)
\right\|_{\mathrm{TV}}
\overset{\mathrm{a.s.}}{\to}
0,
$$
and
$
\sqrt n(\hat\theta_n-\theta_0)
\overset{d}{\to}
N(0,V_{\nu,\theta}).
$
\end{thm}

Such covariance matching is not generic for loss-based generalized
posteriors or quasi-posteriors. For a general empirical loss, the
Bernstein--von Mises covariance is determined by the inverse curvature
$I^{-1}$ of the population criterion, whereas the corresponding
minimum-loss estimator generally has sandwich covariance
$I^{-1}JI^{-1}$
\citep{chernozhukov2003mcmc,miller2021general,kleijn2012bernstein,
muller2013risk}. When $I$ is nonsingular, these two covariance matrices
coincide for the full parameter if and only if the information identity
$J=I$ holds.

For NC-Bayes, covariance matching does not hold for the full parameter
$\gamma=(\theta^\top,\beta)^\top$ as shown in Supplementary Material~\ref{sec:s1}. However, the discrepancy is confined to
the auxiliary normalizing-parameter direction and is therefore eliminated
after marginalization to $\theta$. Hence, exact covariance matching is
recovered for the parameter of interest.

The preceding results concern $n\to\infty$ with $\nu$ fixed.  We now consider
a distinct limit in which $\nu\to\infty$ while the noise distribution $q$
remains fixed. 

\begin{cor}
\label{cor:large-noise-theta-calibration}
Suppose Assumptions~\textnormal{(A1)}, \textnormal{(A2)}, and
\textnormal{(A4)} hold with $r_{\max}=3$, and suppose
Assumption~\textnormal{(A3)} holds for every sufficiently large
$\nu$.  Assume also that $Z(\theta)$ is differentiable
at $\theta_0$.  Let
$
p(x\mid\theta)=\tilde p(x\mid\theta)/Z(\theta)
$
and define
$
\mathcal I(\theta_0)
=
\operatorname{Var}_{p_0}
[
\nabla_\theta\log p(X\mid\theta_0)
].
$
Assume that $\mathcal I(\theta_0)$ is positive definite.  Then
$$
V_{\nu,\theta}
=
R_\theta I_\nu^{-1}R_\theta^\top
\longrightarrow
\mathcal I(\theta_0)^{-1}
\qquad
\text{as }\nu\to\infty.
$$
If, in addition, $\mathbb E_{p_0}[\nabla_\theta \log p(X\mid\theta_0)]=0$, then
$$
\mathcal I(\theta_0)
=
\mathbb E_{p_0}
[
\nabla_\theta\log p(X\mid\theta_0)
\nabla_\theta\log p(X\mid\theta_0)^\top
],
$$
so that $\mathcal I(\theta_0)$ is the Fisher information matrix of the normalized model.
\end{cor}

Thus, if one first considers the fixed-$\nu$ Bernstein--von Mises limit as $n\to\infty$ and then lets $\nu\to\infty$, the covariance matrix $V_{\nu,\theta}$ of the Gaussian approximation to the $\theta$-marginal NC-Bayes posterior converges to
$
\mathcal I(\theta_0)^{-1}.
$
In particular, although $V_{\nu,\theta}$ generally depends on the choice of the noise distribution $q$ for finite $\nu$, this dependence vanishes in the large-noise limit. When $\mathcal I(\theta_0)$ is the Fisher information matrix of the normalized model, the limiting covariance agrees with the inverse Fisher information appearing in the usual likelihood-based Bernstein--von Mises theorem.

\section{Time-varying Density Estimation}
\label{sec:nc-bayes-action1}

\subsection{Model}
Suppose that one is interested in modeling time-varying densities 
$p_t(x)=\exp(f_t(x))/Z_t$ for $t=1,\ldots,T$, where $Z_t=\int \exp(f_t(x)) dx$ is a normalizing constant and $x$ is a $p$-dimensional vector.
For $f_1(x),\ldots,f_T(x)$, we introduce a basis function expansion, $f_t(x)=\sum_{l=1}^L \theta_{tl} \phi_l(x)$, where $\phi_1(x),\ldots,\phi_L(x)$ are basis functions that are common over all the time points. 
For the vector of coefficient parameters, $\theta_t=(\theta_{t1},\ldots,\theta_{tL})^\top$, we put the following random-walk prior:
$$
\theta_t\mid\theta_{t-1}\sim N(\theta_{t-1}, \lambda I_L), \ \ \ t=1,\ldots,T,
$$
with $\theta_0=0$, where $\lambda$ is a scalar parameter that controls the smoothness of the function $f_t(x)$. 
We assign priors for $-\log Z_t$ and $\lambda$ as $\beta_t= -\log Z_t\sim N(0, b_0)$ and $\lambda\sim {\rm IG}(n_0, \nu_0)$ with fixed values $b_0$, $n_0$ and $\nu_0$. 
In the subsequent numerical studies, we set $b_0=10^3$ and $n_0=\nu_0=1$. 

Let $X_t=\{x_{t1},\ldots,x_{tn_t}\}$ be a set of observed data at time $t$. 
Further, let $X_t^{\ast}=\{x_{t,n_t+1},\ldots,x_{t,n_t+m_t}\}$ be a set of noise data generated from a distribution $q(\cdot)$.
Then, the joint posterior of $\Theta=(\theta_1^\top,\ldots,\theta_T^\top)^\top$, $\beta=(\beta_1,\ldots,\beta_T)^\top$, and $\lambda$ given the observed data $X=\{X_1,\ldots,X_T\}$ and the noise data $X^\ast=\{X_1^{\ast},\ldots,X_T^{\ast}\}$ is expressed as 
\begin{align*}
\pi(\Theta, \beta, \lambda \mid X, X^\ast) \propto 
& \ \pi(\lambda, \beta) \prod_{t=1}^T \phi(\theta_t; \theta_{t-1}, \lambda I_L)\prod_{i=1}^{n_t+m_t} \frac{\exp(\Phi_{ti}^\top\theta_t + \beta_t + C_{ti})^{s_{ti}}}{1+\exp(\Phi_{ti}^\top\theta_t + \beta_t + C_{ti})},
\end{align*}
where $s_{ti}=I(i\le n_t)$, $\Phi_{ti}=(\phi_1(x_{ti}),\ldots,\phi_L(x_{ti}))^\top$ and $C_{ti}=\log n_t -\log m_t -\log q(x_{ti})$.
Using Algorithm~\ref{algo:gibbs-hier}, we can develop a Gibbs sampler for posterior computation, whose details are given in Supplementary Material \ref{sec:s2}.

\subsection{Simulation study}

We demonstrate the proposed method using synthetic data.
As the true data-generating distribution, we consider the following two scenarios: 

\begin{itemize}
\item 
(Scenario 1: Time-varying Gaussian mixture) \ \ 
For $t=1,\ldots,T$, each observation $x_{ti}\ (i=1,\ldots,n_t)$ is generated from a two-component Gaussian mixture model with fixed mixing proportions, $0.4N_2(\mu_t^{(1)}, \Sigma^{(1)}) + 0.6N_2(\mu_t^{(2)}, \Sigma^{(2)})$ with time-invariant covariance matrices given by $\Sigma^{(1)}=\mathrm{diag}(0.7,0.2)$ and $\Sigma^{(2)}=0.5 I_2$.
The component means evolve deterministically over time as
$\mu_t^{(1)} = (-2,0)^\top + 4t(1,0)^\top/T$ and $\mu_t^{(2)} = (-2,-2)^\top + 4t(1,1)^\top/T$.

\item 
(Scenario 2: Ring-shaped distribution) \ \ 
For $t=1,\ldots,T$, each observation $x_{ti} \ (i=1,\ldots,n_t)$ is constructed by decomposing it into a radial component $r_{ti}\sim N(\mu_t, \sigma_t^2)$ and an angular component $\theta_{ti}\sim U(0, 2\pi)$.
Then, the two-dimensional observation $x_{ti}$ is generated as $x_{ti} = (r_{ti}\cos\theta_{ti}, r_{ti}\sin\theta_{ti})^\top$.
The parameter $\mu_t$ controls the radius of the ring and increases over time, whereas $\sigma_t$ determines the thickness of the ring and gradually decreases. 
The time-varying parameters $\mu_t$ and $\sigma_t$ are specified deterministically as $\mu_t = 1 + 2(t-1)/(T-1)$ and  $\sigma_t = 0.5 - 0.2(t-1)/(T-1)$.
\end{itemize}
We set $T=10$ and $n_t=100$ in this study.
To illustrate the behavior of the proposed method, we first present a one-shot demonstration using a single realization of the simulated data.
Specifically, we apply NC-Bayes to the generated datasets and visually inspect the resulting density estimates over time.
For NC-Bayes, we draw 3,000 posterior samples after discarding the first 2,000 samples as burn-in.
As the basis functions, we set $L=30$ (the number of basis functions) and employ radial basis functions, $\phi_l(x)=\exp(-\|x-\kappa_l\|/h)$ for $l=1,\ldots,L$, where $\kappa_l$ is the center (knot) and $h$ is a bandwidth. 
Here, the basis centers are determined by applying $k$-means clustering with $L$ clusters to the pooled observations across all time points.
The number of noise samples is set to $m_t=100$, matching the number of observed data points at each time, and we adopt Algorithm~\ref{algo:adaptive_noise} with $\alpha=0.2$ that updates the noise distribution sequentially over time.
As a benchmark method, we apply kernel density estimation (KDE) separately at each time point, without sharing information across time.

For Scenario~2, the true densities and the corresponding estimates obtained by NC-Bayes and KDE are displayed in Figure~\ref{fig:sim-oneshot}.
From these visual comparisons, we observe that NC-Bayes successfully captures both the temporal evolution of the density and the complex, non-Gaussian structure of the distribution.
In contrast, the KDE estimates tend to be overly smooth, producing more diffuse density shapes.
This behavior can be attributed to the fact that KDE is applied independently at each time point and therefore cannot borrow information across different times.
From this perspective, the proposed method benefits from its hierarchical time-series structure, which enables effective sharing of statistical strength across time and leads to more stable and coherent density estimates.

We next conduct a Monte Carlo study with $200$ independent replications to quantitatively assess estimation accuracy and inferential performance.
For each replication, we evaluate the absolute error (ABE) of the estimated density, defined as
$$
{\rm ABE}=T^{-1}\sum_{t=1}^T\int_{D} \left| \widehat{f}_t(x)-f_t(x) \right|dx,
$$
where $D$ denotes the rectangular domain determined by the range of the whole observed data, $\widehat{f}_t(x)$ and $f_t(x)$ are the estimated and true densities at time $t$. 
We approximate this integral via Monte Carlo integration using 500 evaluation points drawn uniformly from $D$. We rescale $\widehat{f}_t$ and $f_t$ so that their discretized integrals over $D$ equal one. To evaluate the uncertainty quantification provided by NC-Bayes, we compute the coverage probability (CP) and average length (AL) of the 95\% credible intervals at the same 500 evaluation points.
For NC-Bayes, we consider three different noise specifications: the adaptive noise distribution described above, a time-invariant common noise distribution shared across all time points (N1), and time-specific noise distributions estimated separately for each time (N2).
The results for both scenarios are summarized in Table~\ref{tab:sim-density}.
Overall, NC-Bayes achieves higher estimation accuracy than KDE in both scenarios, reflecting the advantage of incorporating a hierarchical structure over time.
Regarding inference, all three NC-Bayes variants yield empirical coverage probabilities close to the nominal level.
However, the average length of the credible intervals varies depending on the noise specification, indicating differences in inferential efficiency.
In particular, learning the noise distribution adaptively leads to slightly shorter intervals and improved point estimation accuracy compared with the fixed noise alternatives.

\begin{figure}[t]
\centering
\includegraphics[width=\linewidth]{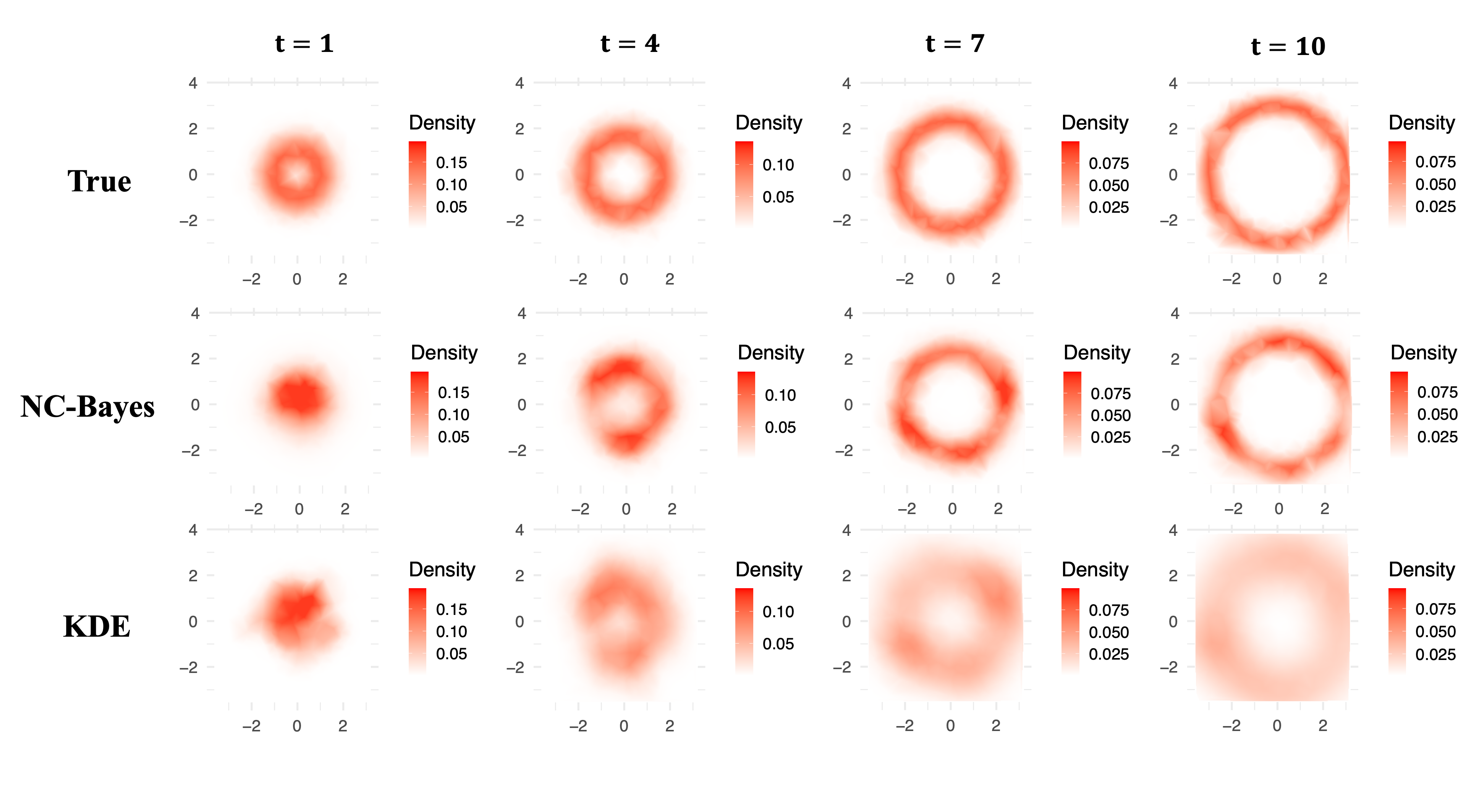}
\caption{ True time-varying density (upper), posterior mean of the time-varying density model fitted by NC-Bayes (middle) and time-wise KDE (lower), for selected time points under Scenario 2. } 
\label{fig:sim-oneshot}
\end{figure}

\begin{table}[t]
\centering
\caption{Absolute error (ABE) of the density estimation, coverage probability (CP) and average length (AL) of pointwise $95\%$ credible intervals, averaged over 500 evaluation points, obtained from time-wise kernel density estimation (KDE) and NC-Bayes with time-invariant noise (N1), time-dependent noise (N2) and adaptive noise (aN).
These values are averaged over 200 Monte Carlo replications. }
\vspace{2mm}
\begin{tabular}{ccccccccccc}
\toprule
&& \multicolumn{4}{c}{Scenario 1} && \multicolumn{4}{c}{Scenario 2} \\
&& \multicolumn{3}{c}{NC-Bayes} &&& \multicolumn{3}{c}{NC-Bayes} &\\
\cmidrule(lr){3-5} \cmidrule(lr){8-10}
 &  & N1 & N2 & aN & KDE &  & N1 & N2 & aN & KDE \\
\midrule
ABE &  & 0.267 & 0.270 & 0.255 & 0.301 &  & 0.350 & 0.375 & 0.371 & 0.581 \\
CP (\%) &  & 96.2 & 95.1 & 93.9 & - &  & 97.2 & 95.8 & 94.5 & - \\
AL &  & 0.154 & 0.130 & 0.113 & - &  & 0.123 & 0.129 & 0.111 & - \\
\bottomrule
\end{tabular}
\label{tab:sim-density}
\end{table}

\subsection{Example: crime incidents in Washington} 

We demonstrate the proposed method using publicly available crime incident data in Washington, DC%
\footnote{available at \url{https://opendata.dc.gov/datasets/DCGIS::crime-incidents-in-2022/about}.}.
The dataset consists of the locations (longitude and latitude) of gun assault incidents observed in each month ($T=12$) in 2022.
The sample size for each month ranges from 60 to 85 observations.
For this dataset, we applied the proposed time-varying density model based on NC-Bayes with adaptive noise updating.
We generated 3,000 posterior samples after discarding the first 2,000 iterations as burn-in, and used the posterior mean as a point estimate.
For comparison, we also fitted kernel density estimators (KDEs) separately for each month.

In Figure~\ref{fig:crime}, we show the estimated density functions for four selected months (January, May, September and December). 
The results indicate that NC-Bayes is able to flexibly capture complex spatial density structures and their temporal evolution over months.
In particular, the proposed method produces sharp and well-localized density estimates while adapting smoothly to changes in the underlying spatial distribution over time.
In contrast, the KDE results tend to be overly smooth across all months, which is likely due to the relatively small sample size available in each period.
This behavior is consistent with the findings from our simulation studies, where KDE exhibited limited ability to recover complex or rapidly changing density shapes under small-sample settings.

\begin{figure}[t]
\centering
\includegraphics[width=\linewidth]{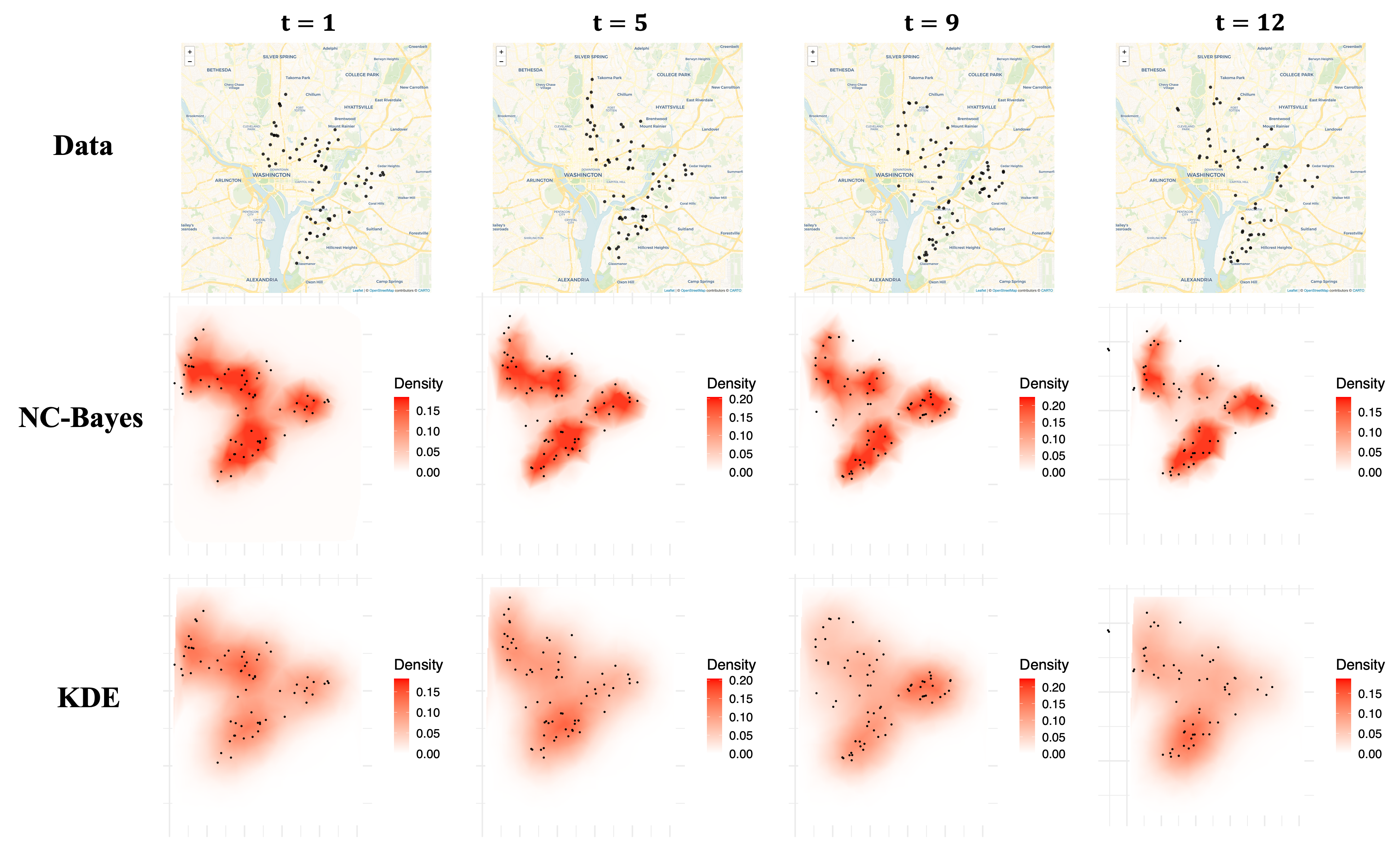}
\caption{Estimated spatial density functions for gun assault incident locations in Washington, DC, for four selected months ($t=1,5,9,12$) in 2022. The observed points (upper) and the posterior mean of the time-varying density model fitted by NC-Bayes (middle) are compared with the time-wise KDE (lower).} 
\label{fig:crime}
\end{figure}

\section{Sparse Torus Graph}
\label{sec:nc-bayes-action2}

\subsection{Model}
\label{sec:nc-bayes-action2.1}
Angular measurements are naturally modeled as circular random variables, whose joint distribution for $d$ variables lies on the $d$-dimensional torus, the product of $d$ circles. To describe dependencies among such variables, \cite{klein2020torus} introduced torus graphs, a class of regular full exponential families with pair interactions. 

Let $x=(x_1, \dots, x_d)^\top$ be a vector of circular variables, each $x_j \in [0, 2\pi)$. The torus graph density can be expressed as 
\begin{equation}\label{model_torus_graph}
p(x\mid \phi) \propto 
\pt(x\mid \phi) =
\exp \left\{ \sum_{j=1}^{d} \phi_j^\top \theta(x_j)+\sum_{j<k} \phi_{jk}^\top \eta(x_j,x_k) \right\},
\end{equation}
where $\theta(x_j)=(\cos x_j, \sin x_j)^\top$,~$\eta(x_j, x_k)=(\cos (x_j - x_k), \sin (x_j - x_k), \cos (x_j + x_k), \sin (x_j + x_k))^\top$ and $\phi=(\phi_1^\top,\ldots,\phi_d^\top,\phi_{12}^\top,\ldots,\phi_{d-1,d}^\top)^\top\in \mathbb{R}^{2d^2}$. 
Here, the marginal parameters $\phi_j=(\phi_{j(1)}, \phi_{j(2)})^\top\in \mathbb{R}^2$ control the first circular moments of each variable, while the coupling parameters $\phi_{jk}=(\phi_{jk(1)}, \phi_{jk(2)}, \phi_{jk(3)}, \phi_{jk(4)})^\top\in \mathbb{R}^4$ govern pairwise dependence. The first two components of $\phi_{jk}$ correspond to rotational dependence, and the last two correspond to reflectional dependence. As in Gaussian graphical models, zero coupling parameters $\phi_{jk}=0$ imply that $x_j$ and $x_k$ are conditionally independent given all other variables, enabling a direct graphical model interpretation. However, a major computational challenge of the model~(\ref{model_torus_graph}) is that its normalizing constant is analytically intractable, making standard maximum likelihood estimation and Bayesian inference computationally prohibitive. These computational challenges motivate the use of NC-Bayes.

%
Let $X_n=\{x_1,\ldots,x_n\}$ be observed data and $X_{m}^\ast=\{x_{n+1},\ldots,x_{n+m}\}$ be artificially generated noise data from some noise distribution $q(\cdot)$, where $x_i = (x_{i1},\dots, x_{id})^\top$ for $i=1, \dots, n+m$. 
For notational convenience, let $z(x_i)=(\theta(x_{i1})^\top,\ldots,\theta(x_{id})^\top, \eta(x_{i1},x_{i2})^\top,\\\ldots,\eta(x_{i,d-1},x_{id})^\top, 1)^\top$ and $\gamma=(\phi_1^\top,\ldots,\phi_d^\top ,\phi_{12}^\top, \ldots, \phi_{d-1,d}^\top, \beta)^\top$, where $\beta = -\log Z$ with $Z$ being the normalizing constant. Then the prior for $\gamma$ is $N(b_{\gamma 0}, B_{\gamma 0})$, where $b_{\gamma 0}=(b_{\phi_10}^\top,\ldots,b_{\phi_d0}^\top, b_{\phi_{12} 0}^\top,\ldots,b_{\phi_{d-1, d} 0}^\top, b_{\beta 0})^\top$ and $B_{\gamma 0}={\rm diag}(B_{\phi_10},\ldots,B_{\phi_d0}, B_{\phi_{12} 0},\ldots,B_{\phi_{d-1, d} 0}, \\B_{\beta0})$. 
Since the torus graph model belongs to the exponential family, the posterior distribution of $\gamma$ can be written in the form of equation~\eqref{eq:pos-noise}, and the data augmentation scheme introduced in Section \ref{sec:Gibbs} applies directly, so that Algorithm \ref{algo:gibbs2} can be used without modification.
Consequently, the full conditional distributions of $\phi_j$, $\phi_{jk}$, and $\beta$ are all Gaussian given the Gaussian priors specified above. The complete sampling algorithm based on this augmentation is detailed in Supplementary Material \ref{sec:comp-nc-bayes-torus}, as the following section extends this framework by introducing shrinkage priors on $\phi_j$ and $\phi_{jk}$.

\subsection{Handling sparsity via horseshoe priors}
\label{sec:nc-bayes-action2.2}

In high-dimensional settings, the naive estimation of the torus graph model often results in densely connected graphs \citep{klein2020torus}, which can obscure the underlying meaningful network structure. Previous studies have addressed this issue by applying multiple-testing procedures to dense estimators \citep{klein2020torus} or by introducing regularization such as the group lasso and an information criterion into score matching to induce sparsity  \citep{matsuda2021information,sukeda2025Directional}. In the present work, we take a fully Bayesian approach and introduce sparsity in $\phi$ through shrinkage priors. Specifically, we can introduce the horseshoe prior \citep{carvalho2009handling,carvalho2010horseshoe} defined as $\phi_{j(l)}\sim N(0, \lambda_{j(l)}^2 \tau^2)$ for $l=1, 2$ and $\phi_{jk(l')}\sim N(0, \xi_{jk(l')}^2\tau^2)$ for $l'=1, \dots, 4$, where $\lambda_{j(l)}$ and $\xi_{jk(l')}$ are local shrinkage parameters following the half-Cauchy distribution $C^+(0,1)$, and $\tau^2$ is a common global variance. This formulation applies shrinkage individually to each component of $\phi_{jk}$, allowing weak edges to be strongly shrunk toward zero while retaining truly strong connections. Moreover, to impose sparsity on the entire set of interaction parameters between nodes $j$ and $k$ in a coherent manner, we consider a grouped horseshoe prior. In this setting, the four components of $\phi_{jk}$ share a group-level shrinkage parameter $u_{jk}$, namely $\phi_{jk(l')}\sim N(0, u_{jk}^2 \tau^2)$ for $l'=1,\dots,4$. The group-level parameter $u_{jk}$ controls the overall relevance of the edge $(j,k)$. This structure enables sparsity induction at the edge level, leading to more interpretable and parsimonious network estimates.

However, posterior inference for $\phi$ is driven by a logistic likelihood through NCE, so high-dimensional shrinkage modeling can inherit the same pathology as sparse logistic regression \citep{nishimura2023shrink}. In particular, weak identifiability combined with heavy-tailed global-local priors can yield heavy-tailed posteriors \citep{nishimura2023shrink}. This can undermine geometric ergodicity of the P\'olya--Gamma Gibbs sampler, leading to poor mixing and overly wide credible intervals \citep{nishimura2023shrink}.
In our approach, we mitigate this instability by replacing the grouped horseshoe with its regularized counterpart, introducing a finite slab width $c>0$ \citep{piironen2017shrink}. 
Concretely, for the node-specific parameters we set, for $l=1,2$,
\begin{equation*}
\phi_{j(l)} \sim
N\left(0, \left(\frac{1}{c^{2}}+\frac{1}{\lambda_{j(l)}^{2}\tau^{2}}\right)^{-1}\right),
~
\lambda_{j(l)} \sim \mathrm{C}^+(0,1),
\end{equation*}
and for the edge-level components we regularize the grouped horseshoe by placing the slab at the group scale:
for $l'=1,\dots,4$,
\begin{equation*}
\phi_{jk(l')} \sim
N\left(0, \left(\frac{1}{c^{2}}+\frac{1}{u_{jk}^{2}\tau^{2}}\right)^{-1}\right),
~
u_{jk} \sim \mathrm{C}^+(0,1).
\end{equation*}
With this approach, we enforce an admissible network structure through a constrained parameterization of $\phi$, while retaining edge-wise shrinkage via the group scales $u_{jk}$. In the following sections, we use the fixed slab width $c = 1$.

\subsection{Simulation study}
\label{sec:nc-bayes-action2.3}

We generate synthetic data following the simulation scenario of \cite{klein2020torus}.
We set $n=200$ and $d=12$, and generate
$x_1$ from a von-Mises distribution with mean direction $\mu=\pi/6$ and concentration parameter $\kappa=2$, and
$x_j$ from a von-Mises distribution with mean direction $x_{j-1}+\mu$ and concentration parameter $\kappa=2$ for $j=2,\ldots,d$.
The resulting true graph is a linear chain in which node $j$ is connected to node $j+1$ for $j=1,\ldots,d-1$.
We apply NC-Bayes with the regularized grouped horseshoe prior.
Based on 2,000 posterior samples after discarding the first 1,000 samples as burn-in, an edge $(j,k)$ is detected when at least one of
$\{|\tilde{\phi}_{jk(1)}|,\ldots,|\tilde{\phi}_{jk(4)}|\}$
exceeds 0.100, where $\tilde{\phi}_{jk(l)}$ denotes the posterior median of $\phi_{jk(l)}$.
Figure~\ref{fig:oneshot} shows the detected edges, with edge width proportional to
$\sum_{l=1}^4|\tilde{\phi}_{jk(l)}|$.
NC-Bayes successfully recovers all edges in the true linear chain without detecting spurious edges.

\begin{figure}[t]
\centering
\includegraphics[width=0.7\linewidth]{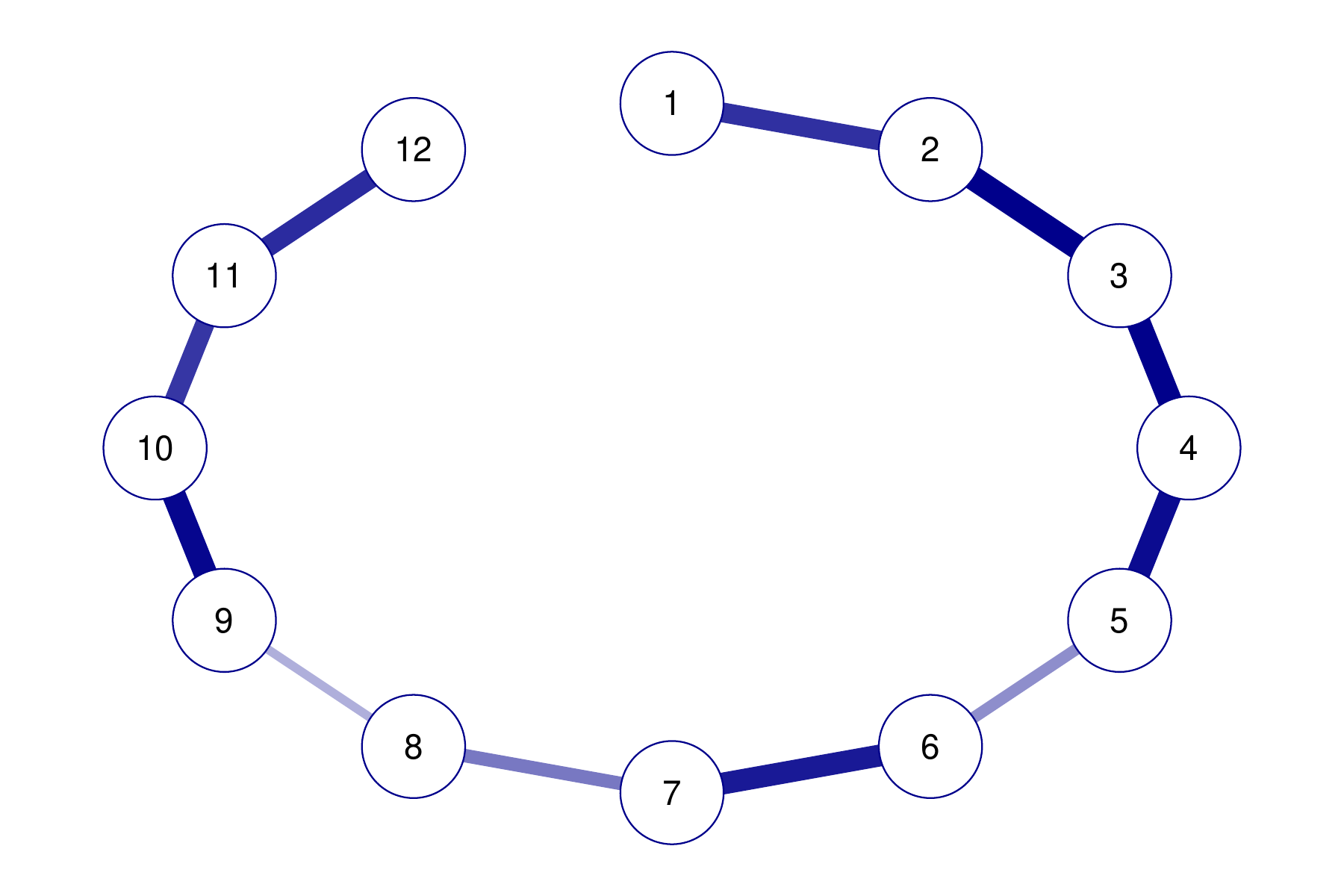}
\caption{
Detected edges from NC-Bayes.
The true graph is a linear chain in which node $j$ is connected to node $j+1$ for $j=1,\ldots,d-1$.
}
\label{fig:oneshot}
\end{figure}

We further compare NC-Bayes with a Hyv\"arinen score-based generalized Bayesian method
\citep{jewson2022general,altamirano2023score,altamirano2024robust},
which we refer to as $\mathcal{H}$-Bayes.
Details of its generalized posterior and sampling algorithm are provided in Supplementary Material~\ref{sec:comp-h-bayes}.
For both methods, we use a grouped horseshoe prior, with its regularized counterpart for NC-Bayes.
For NC-Bayes, we consider two choices for the number of noise samples: the same as the observed sample size and five times the observed sample size, corresponding to $\nu=1$ and $\nu=5$, respectively.
We also consider NC-Bayes both with and without adaptive noise updating; the former uses Algorithm~\ref{algo:adaptive_noise} with $\alpha=0.1$.
For $\mathcal{H}$-Bayes, we examine the sensitivity to the learning-rate parameter $w$ by considering $w\in\{0.2,0.5,1.0\}$.
The results are averaged over 100 simulations.

Tables~\ref{tab:results_Median} and~\ref{tab:results_CI} report the results under the median-based and 90\% credible-interval (CI) edge detection rules, respectively.
Under the median-based rule, NC-Bayes performs well when $\nu=1$, achieving high recall, precision, and accuracy both with and without adaptive noise updating.
Under the CI-based rule, however, when $\nu=1$ and the noise distribution is fixed, recall is relatively low, whereas adaptive noise updating substantially improves recall while maintaining near-perfect precision and high coverage probability.
When the number of noise samples is increased to $\nu=5$, the CI-based performance improves markedly even without adaptive noise updating, with recall approaching one while retaining near-perfect precision and high coverage.
Although these edge-selection metrics do not directly assess the asymptotic variance, the improved CI-based behavior with a larger number of noise samples is in line with the asymptotic result in Corollary~\ref{cor:large-noise-theta-calibration}, which shows that the asymptotic variance approaches that determined by the Fisher information as the amount of noise data increases.
In contrast, the performance of $\mathcal{H}$-Bayes remains sensitive to the loss-scaling parameter $w$: larger values of $w$ lead to substantially more false-positive detections and poorer interval coverage.
Overall, these results demonstrate that NC-Bayes provides stable edge-selection performance under the threshold-based rule, while its uncertainty quantification under the CI-based rule improves either through adaptive noise updating or by increasing the number of noise samples.
Additional simulation results under alternative scenarios are provided in Supplementary Material~\ref{sec:simulation-torus-graph}, where similar results are consistently observed.

\begin{table}[t]
\centering
\caption{
Performance metrics averaged over 100 simulations under the median-based edge detection rule.
An edge is detected when the absolute posterior median of any associated coefficient exceeds 0.100.
}
\vspace{2mm}
\begin{tabular}{llllrrr}
\hline
Method & $\nu$ & $w$ & Noise Update & Recall & Precision & Accuracy \\
\hline\hline
\multirow{4}{*}{NC-Bayes}
& \multirow{2}{*}{1} & -- & False & 0.996 & 0.999 & 0.999 \\
&                        & -- & True  & 1.000 & 0.985 & 0.997 \\
\cline{2-7}
& \multirow{2}{*}{5} & -- & False & 1.000 & 0.998 & 1.000 \\
&                        & -- & True  & 1.000 & 0.874 & 0.976 \\
\hline
\multirow{3}{*}{$\mathcal{H}$-Bayes}
& -- & 0.2 & -- & 1.000 & 0.699 & 0.928 \\
& -- & 0.5 & -- & 1.000 & 0.198 & 0.323 \\
& -- & 1.0 & -- & 1.000 & 0.174 & 0.207 \\
\hline
\end{tabular}
\label{tab:results_Median}
\end{table}

\begin{table}[t]
\centering
\caption{
Performance metrics averaged over 100 simulations under the 90\% CI-based edge detection rule.
An edge is detected when its 90\% credible interval does not contain zero.
CP denotes the coverage probability for $\phi$.
}
\vspace{2mm}
\begin{tabular}{llllrrrr}
\hline
Method & $\nu$ & $w$ & Noise Update & CP (\%) & Recall & Precision & Accuracy \\
\hline\hline
\multirow{4}{*}{NC-Bayes}
& \multirow{2}{*}{1} & -- & False & 97.9 & 0.507 & 1.000 & 0.918 \\
&                        & -- & True  & 98.6 & 0.734 & 1.000 & 0.956 \\
\cline{2-8}
& \multirow{2}{*}{5} & -- & False & 98.5 & 0.985 & 1.000 & 0.998 \\
&                        & -- & True  & 99.0 & 0.999 & 0.999 & 1.000 \\
\hline
\multirow{3}{*}{$\mathcal{H}$-Bayes}
& -- & 0.2 & -- & 97.2 & 1.000 & 0.941 & 0.990 \\
& -- & 0.5 & -- & 70.4 & 1.000 & 0.246 & 0.490 \\
& -- & 1.0 & -- & 45.1 & 1.000 & 0.182 & 0.250 \\
\hline
\end{tabular}
\label{tab:results_CI}
\end{table}

\subsection{Example: multichannel neural phase angle data}

We analyze multichannel neural phase data from a macaque monkey experiment previously studied in
\cite{brincat2015neuro,brincat2016neuro,klein2020torus}.
The data consist of beta-phase values at 300 ms after stimulus onset across 840 trials for 24 simultaneously recorded signals from the prefrontal cortex (PFC) and hippocampal regions, including the dentate gyrus (DG), CA3, and subiculum (Sub).
We use the preprocessed data\footnote{Available at \url{https://github.com/natalieklein/torus-graphs}.} obtained by applying Morlet wavelets to the LFP signals \citep{klein2020torus}.
Figure~\ref{fig:overview-real-data} shows pairwise phase relationships for representative channels from the four regions.

\begin{figure}[t]
\centering
\includegraphics[width=0.9\linewidth]{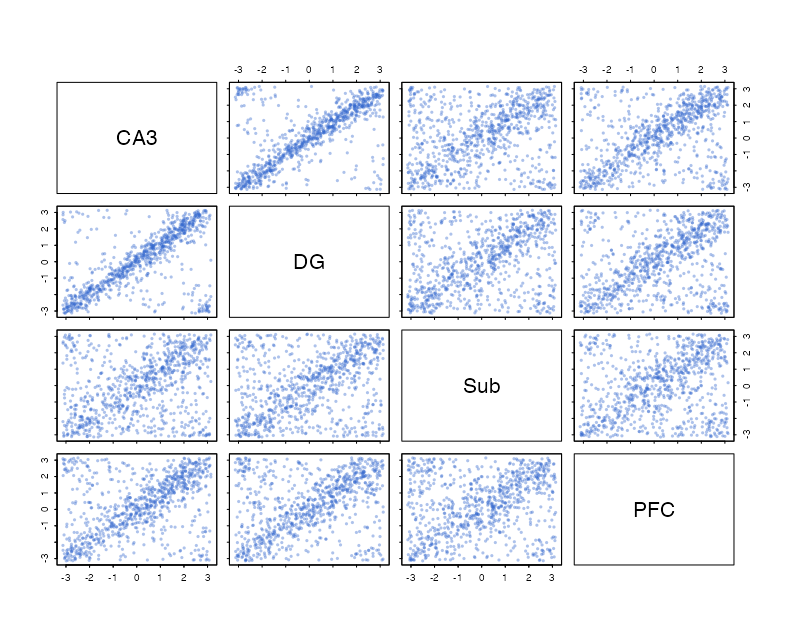}
\caption{
Pairwise scatter plots of beta-phase values for one selected channel from each of CA3, DG, Sub, and PFC.
Each axis ranges from $-\pi$ to $\pi$.
}
\label{fig:overview-real-data}
\end{figure}

We fit a 24-dimensional torus graph model using NC-Bayes and, for comparison, $\mathcal{H}$-Bayes as in Section~\ref{sec:nc-bayes-action2.3}.
Both methods use grouped horseshoe priors, with the regularized version for NC-Bayes.
For NC-Bayes, the global shrinkage scale is fixed at
$\tau=p_0/\{\sqrt{n+m}(2d^2-p_0)\}$, where
$p_0=\lfloor1.7d^2+0.5\rfloor$, following \cite{piironen2017shrink},
and the noise distribution is adaptively updated during burn-in using
Algorithm~\ref{algo:adaptive_noise} with $\alpha=0.1$.
For both methods, posterior inference is based on 2,000 samples after discarding the first 1,000 as burn-in.
An edge is detected when the absolute posterior median of at least one associated coefficient exceeds 0.020.

Figure~\ref{fig:oneshot-real-data-NC-Bayes} shows the network estimated by NC-Bayes.
The inferred graph identifies direct PFC--CA3 and PFC--Sub connections, but no direct PFC--DG connection, consistent with the connectivity pattern reported by \cite{klein2020torus}.
DG is instead connected to CA3 and Sub.
In contrast, $\mathcal{H}$-Bayes with $w=0.2$, which yielded the sparsest estimates among the settings considered in the simulation study, produces a substantially denser graph (Figure~\ref{fig:oneshot-real-data-Hscore}).

\begin{figure}[t!]
\centering
\includegraphics[width=0.9\linewidth]{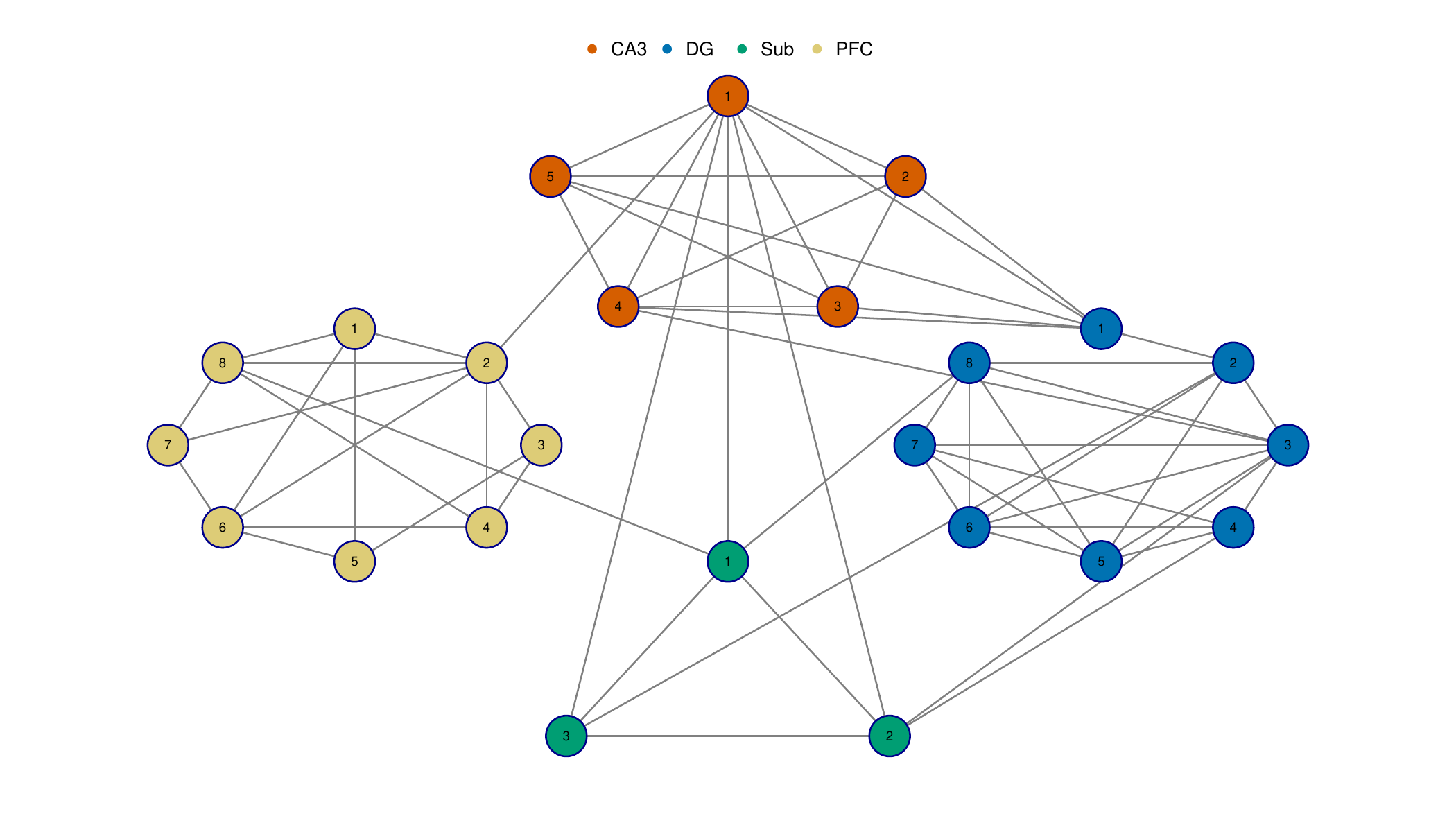}
\caption{Inferred torus graph structure from the phase-angle data using NC-Bayes.}
\label{fig:oneshot-real-data-NC-Bayes}
\end{figure}

\begin{figure}[t!]
\centering
\includegraphics[width=0.9\linewidth]{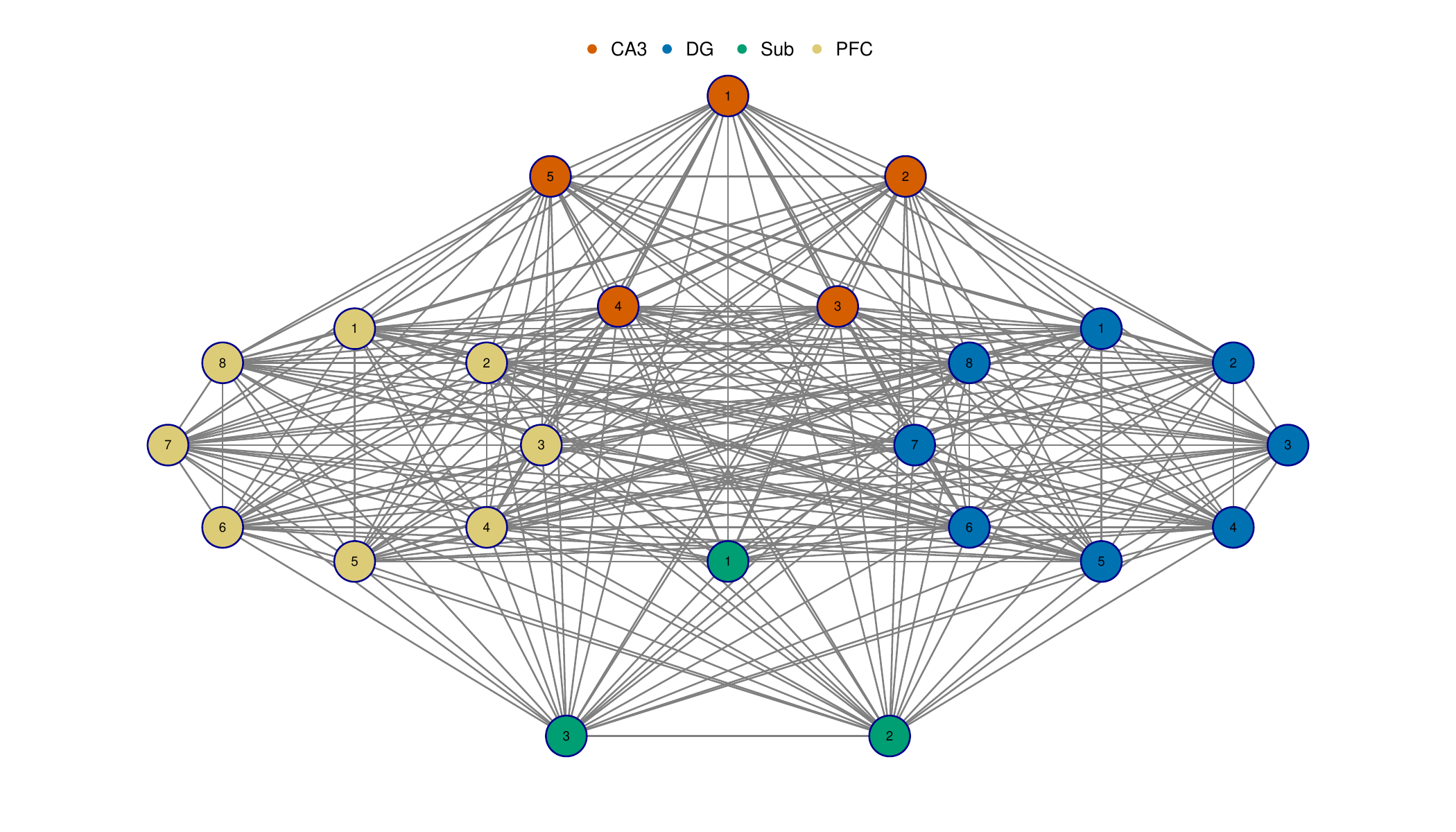}
\caption{Inferred torus graph structure from the phase-angle data using $\mathcal{H}$-Bayes with $w=0.2$.}
\label{fig:oneshot-real-data-Hscore}
\end{figure}

We further assess uncertainty in edge presence using posterior medians and 50\% credible intervals (Figure~\ref{fig:posterior-median-ci}).
For NC-Bayes, several edges selected based on the posterior median have credible intervals containing zero, indicating that the evidence for their presence becomes insufficient once posterior uncertainty is taken into account.
In contrast, for $\mathcal{H}$-Bayes, edge selections based on the median and credible intervals largely coincide, while the intervals become substantially narrower as $w$ increases.
Given that larger values of $w$ tend to produce denser estimated networks in our simulation studies, this behavior suggests that $\mathcal{H}$-Bayes may exhibit high posterior confidence even for potentially spurious edges.

\begin{figure}[t!]
\centering
\begin{tabular}{cc}
\includegraphics[width=0.5\linewidth]{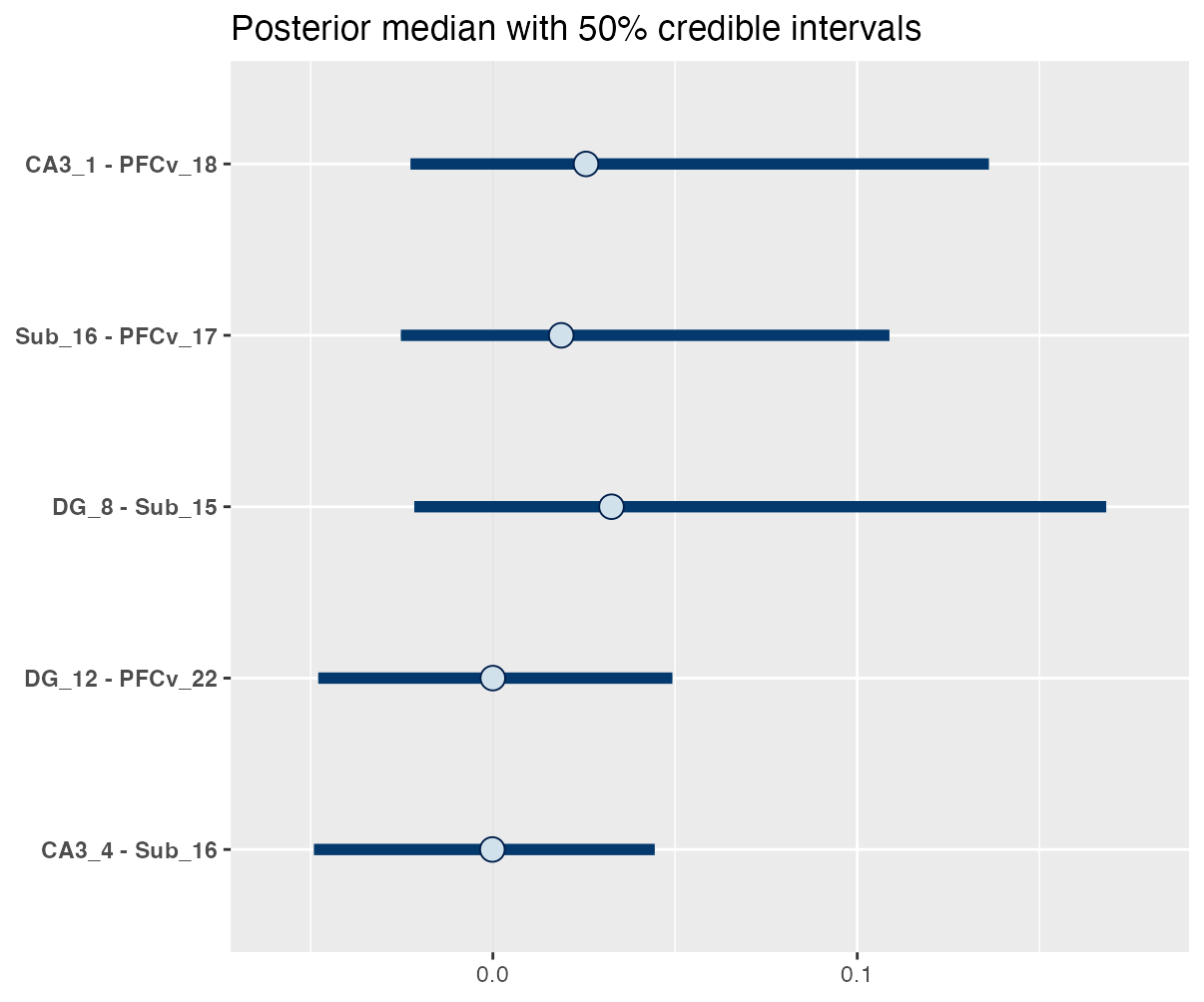} &
\includegraphics[width=0.5\linewidth]{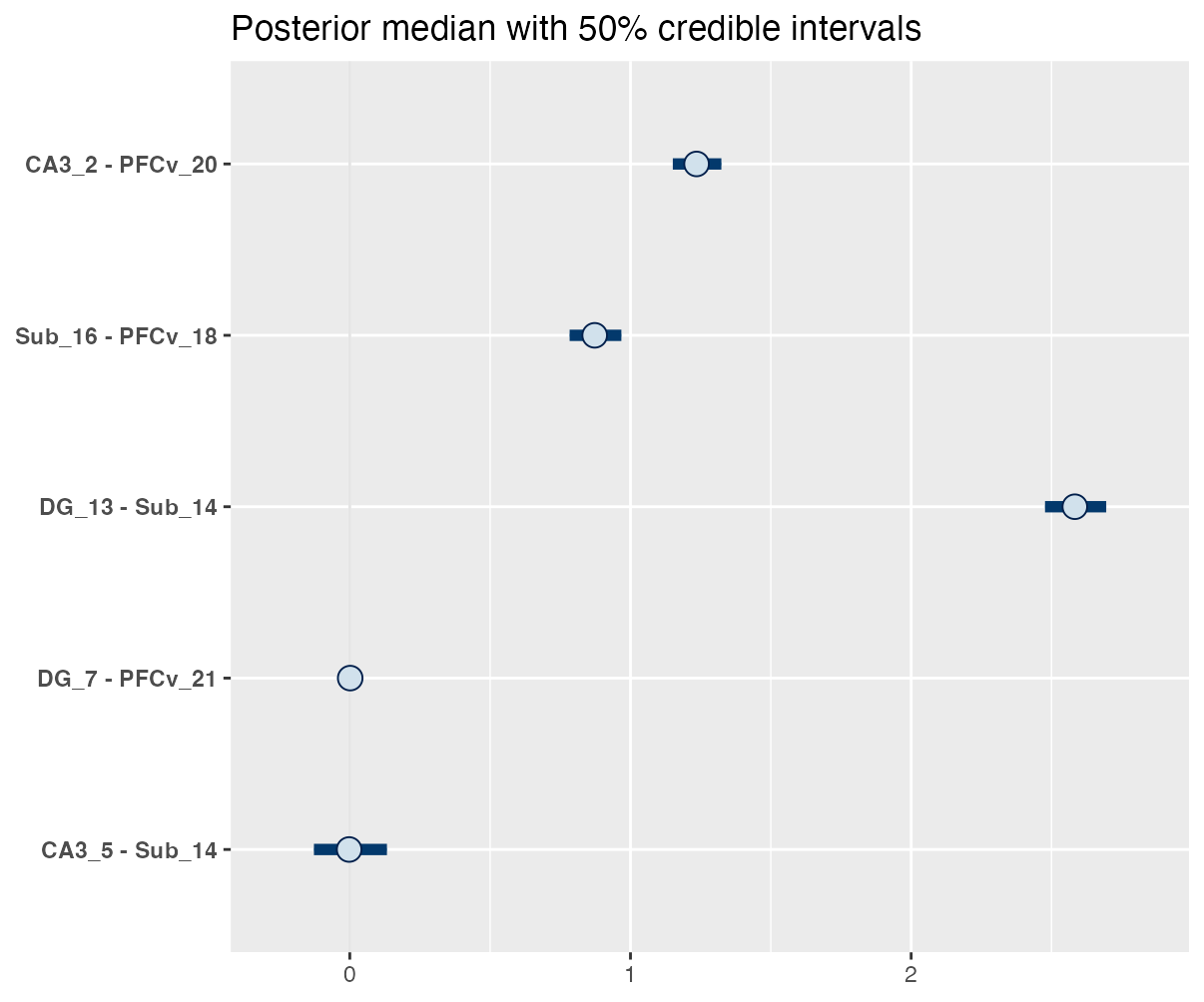} \\
\includegraphics[width=0.5\linewidth]{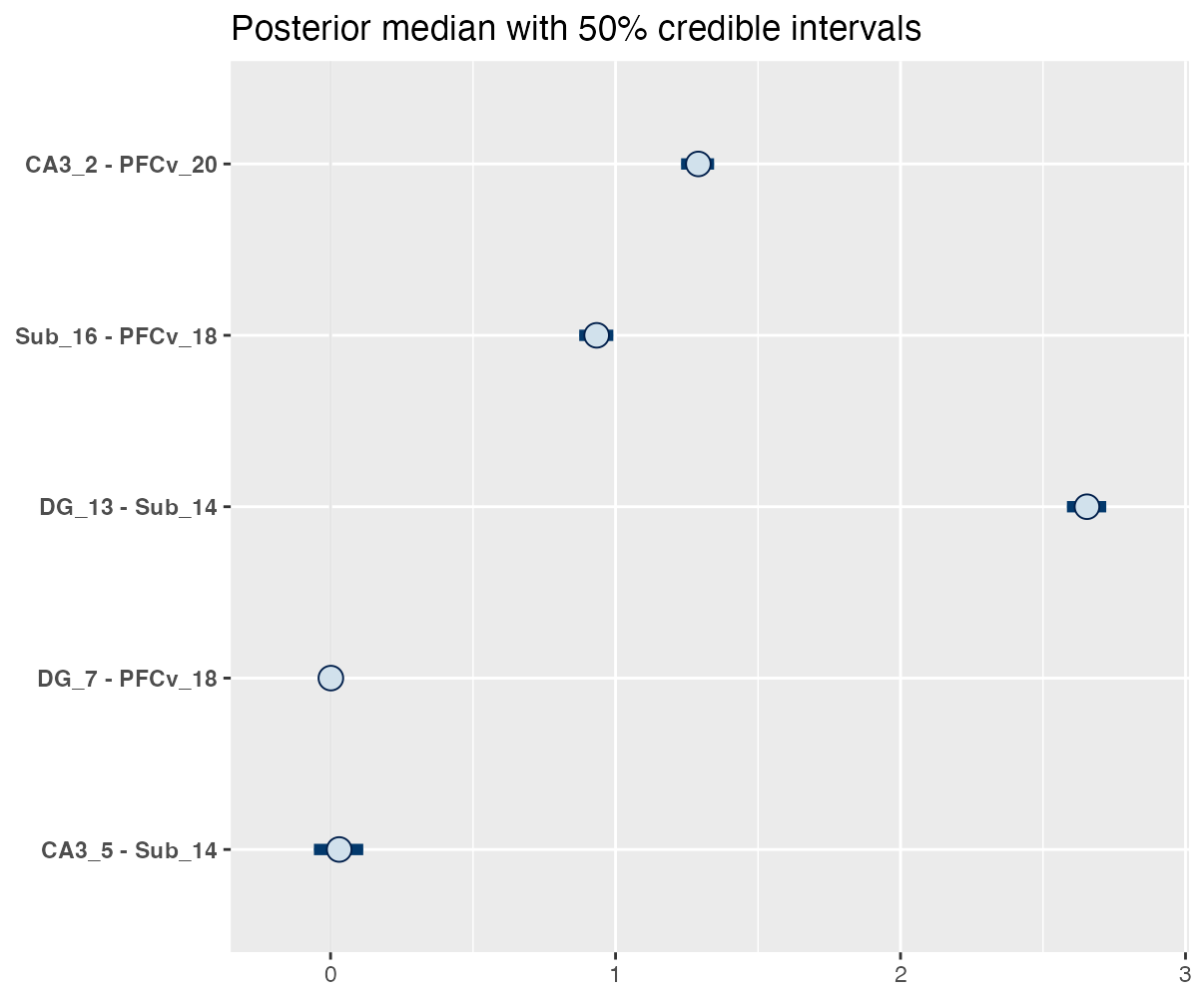} &
\includegraphics[width=0.5\linewidth]{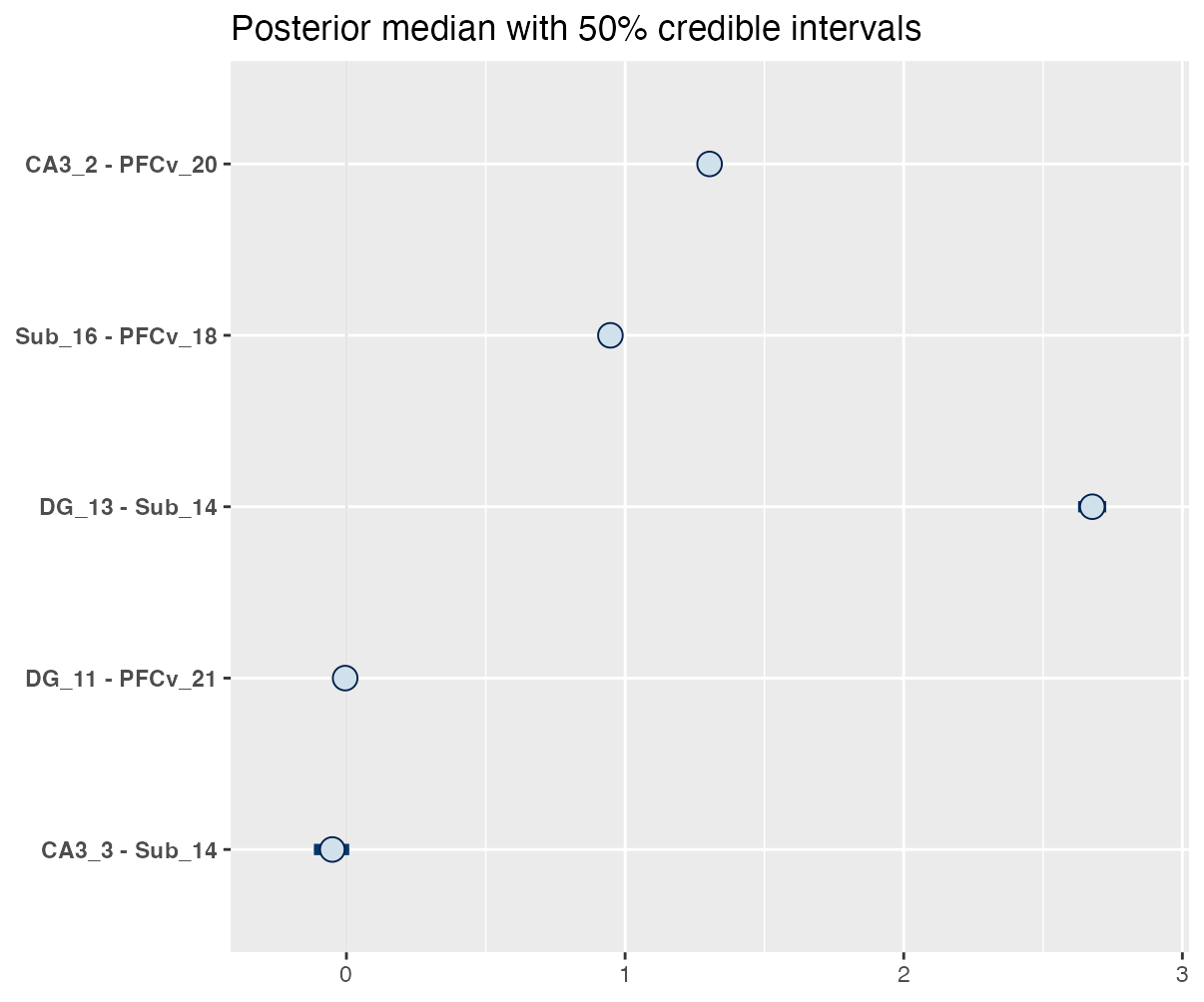} 
\end{tabular}
\caption{
Posterior medians and 50\% credible intervals for selected edge coefficients.
The panels show NC-Bayes (top left) and $\mathcal{H}$-Bayes with
$w=0.2$ (top right), $0.5$ (bottom left), and $1$ (bottom right).
Circles denote posterior medians and horizontal segments denote 50\% credible intervals.
The top three examples correspond to region pairs identified as associated in \cite{klein2020torus}, and the bottom two to pairs identified as not associated.
}
\label{fig:posterior-median-ci}
\end{figure}

\section{Discussion}
\label{sec:discussion}
While the main text focuses on the MCMC-based perspective of NC-Bayes, a further connection to likelihood-free inference is worth discussing. A growing body of work has explored classification-based approaches, including classifier-based ABC \citep{gutmann2018classification, wang2022approximate} and methods that plug likelihood-ratio estimators into the Metropolis--Hastings algorithm \citep{kaji2023metropolis}. These approaches share with NC-Bayes the motivation of avoiding direct evaluation of an intractable quantity, though they primarily target simulator-based models where the likelihood itself is unavailable, whereas NC-Bayes targets unnormalized models where the likelihood is known up to a normalizing constant. Unlike ABC, which introduces approximation error through summary statistics and tolerance thresholds, NC-Bayes requires neither, yielding a posterior free from these sources of approximation. Moreover, the P\'olya--Gamma data augmentation yields closed-form Gibbs conditionals, an analytical tractability that these likelihood-free approaches do not generally offer.

Despite these advantages, NC-Bayes is not without limitations. Since the NC-Bayes approach involves a logistic regression likelihood, it becomes considerably challenging in high-dimensional settings. In Section \ref{sec:nc-bayes-action2.2}, we introduced a method using the regularized horseshoe to address a problem that arises in high-dimensional logistic regression: regression coefficients identified as signals can take on extremely large values. This issue is mitigated by restricting the tail weight, and we employed this approach in our simulations and data analysis. However, reducing the tail weight creates the problem that the distribution struggles to escape from the neighborhood of zero, which may lead to excessively sparse estimates. Therefore, selecting an appropriate prior distribution that enables proper shrinkage estimation in high-dimensional logistic regression remains an important challenge.

A key direction for future work is the development of methods for selecting the noise distribution. 
While we have considered fixed noise distributions and practical adaptation schemes within the MCMC framework in this paper, a more systematic strategy is needed. According to \cite{chehab2022noise}, the theoretically optimal noise distribution need not coincide with the data distribution in the frequentist framework. Similarly, deriving a theoretically optimal noise selection method for the NC-Bayes approach remains an important direction. 
Additionally, while we estimate the data density during MCMC using importance resampling, an alternative approach would be to pre-estimate the data distribution using a nonparametric high-precision density estimator and then use it as the noise distribution \citep{uehara2020unnormalized}.

\section*{Acknowledgements}
This work is partially supported by the Japan Society for the Promotion of Science (KAKENHI) grant numbers 24K21420 and 25H00546. 
Naruki Sonobe is supported by JST Grant Number JPMJSP2123.
Takeru Matsuda was supported by JSPS KAKENHI Grant Numbers
21H05205, 22K17865 and 24K02951 and JST Grant Numbers JPMJMS2024 and JPMJAP25B1.

\vspace{1cm}
\bibliographystyle{chicago}
\bibliography{ref}

\newpage
\setcounter{equation}{0}
\setcounter{section}{0}
\setcounter{table}{0}
\setcounter{figure}{0}
\setcounter{page}{1}
\renewcommand{\thesection}{S\arabic{section}}
\renewcommand{\theequation}{S\arabic{equation}}
\renewcommand{\thetable}{S\arabic{table}}
\renewcommand{\thefigure}{S\arabic{figure}}

\vspace{1cm}
\begin{center}
{\LARGE
{\bf Supplementary Material for ``Contrastive Bayesian Inference for Unnormalized Models''}
}
\end{center}

This Supplementary Material provides proofs, details of the posterior sampling algorithms and additional simulation results.

\section{Proofs for the asymptotic properties of NC-Bayes}
\label{sec:s1}

This section establishes posterior concentration and a Bernstein--von Mises
theorem for the NC-Bayes posterior.  The proof strategy follows the
generalized-posterior analysis of \citet{matsubara2022robust}: uniform
convergence of the empirical loss is combined with a directly imposed
well-separatedness condition, while the Bernstein--von Mises theorem is
obtained by verifying the conditions of \citet{miller2021general}.  The
large-noise limit is treated separately and formalizes the connection between
noise-contrastive estimation and ordinary likelihood inference described by
\citet{gutmann2012noise}.

\subsection{Notation and assumptions}

Let $\lambda$ be a dominating $\sigma$-finite measure and define
$\mathcal X_0=\{x:p_0(x)>0\}$.  We assume that $q(x)>0$ for
$\lambda$-almost every $x\in\mathcal X_0$.  Thus, the noise distribution
may have support strictly larger than the support of the data-generating
distribution.
For every $\theta$ under consideration, assume that the unnormalized model
has the common support $\mathcal X_0$, in the sense that
$\tilde p(x\mid\theta)>0$ for $x\in\mathcal X_0$ and
$\tilde p(x\mid\theta)=0$ for $x\notin\mathcal X_0$.  Assume also that the
model is correctly specified at
$\gamma_0=(\theta_0^\top,\beta_0)^\top$:
$p_0(x)=e^{\beta_0}\tilde p(x\mid\theta_0)$ for
$\lambda$-almost every $x$.
Let $R_\theta=(I_p,0)\in\mathbb R^{p\times(p+1)}$ denote the coordinate projection from $\gamma=(\theta^\top,\beta)^\top$ to $\theta$, so that $R_\theta\gamma=\theta$ and $R_\theta\gamma_0=\theta_0$. For any measurable set $A\subseteq\mathbb R^p$, let
$
R_\theta^{-1}(A)=\{\gamma\in\mathbb R^{p+1}:R_\theta\gamma\in A\}
$
denote its inverse image under $R_\theta$.
Fix $\nu\in\mathbb Q_{>0}$.  We consider sample sizes $n$ for which
$m=\nu n\in\mathbb N$.  Let
$P_n=n^{-1}\sum_{i=1}^n\delta_{x_i}$ and
$Q_m=m^{-1}\sum_{i=n+1}^{n+m}\delta_{x_i}$, where
$x_1,\ldots,x_n\overset{\mathrm{i.i.d.}}\sim p_0$ and
$x_{n+1},\ldots,x_{n+m}\overset{\mathrm{i.i.d.}}\sim q$ independently.
We use the same notation for a probability measure and its associated
integration operator.  Thus, for any integrable measurable function $h$,
$P_nh=n^{-1}\sum_{i=1}^n h(x_i)$,
$Q_mh=m^{-1}\sum_{i=n+1}^{n+m}h(x_i)$,
$P_0h=\mathbb E_{p_0}[h(X)]$, and
$Qh=\mathbb E_q[h(X^\ast)]$.
For $x\in\mathcal X_0$ and
$\gamma=(\theta^\top,\beta)^\top$, define
$a(x\mid\gamma)=\log\tilde p(x\mid\theta)+\beta-\log q(x)$ and
$\rho_\nu(x\mid\gamma)=\{1+\nu e^{-a(x\mid\gamma)}\}^{-1}$.  Set
$\rho_\nu(x\mid\gamma)=0$ for $x\notin\mathcal X_0$.  Define
$\phi(x\mid\gamma)=\log\{1+\nu e^{-a(x\mid\gamma)}\}$ on
$\mathcal X_0$, with an arbitrary finite value on
$\mathcal X_0^c$, and define
$\psi(x\mid\gamma)=\log\{1+\nu^{-1}e^{a(x\mid\gamma)}\}$ on
$\mathcal X_0$ and $\psi(x\mid\gamma)=0$ on
$\mathcal X_0^c$.  Then
$f_{n,\nu}(\gamma)=P_n\phi(\cdot\mid\gamma)+\nu Q_m\psi(\cdot\mid\gamma)$
and
$f_\nu(\gamma)=P_0\phi(\cdot\mid\gamma)+\nu Q\psi(\cdot\mid\gamma)$.
Write $\rho_{0, \nu}(x)=\rho_\nu(x\mid\gamma_0)$ and
$s(x)=\nabla_\gamma a(x\mid\gamma)|_{\gamma=\gamma_0}$ on
$\mathcal X_0$, and set $s(x)=0$ on $\mathcal X_0^c$.

\begin{as}
\label{as:nc-bayes}
The following conditions are imposed.
\begin{enumerate}

\item[(A1)]
The parameter set
$\Gamma\subset\mathbb R^{p+1}$ is open, bounded and convex,
$\gamma_0\in\Gamma$, and
$R_\theta(\Gamma)\subset\Theta$.

\item[(A2)]
For a given $r_{\max}\in\{1,3\}$ and every integer
$0\le r\le r_{\max}$, the map
$\gamma\mapsto\nabla_\gamma^r a(x\mid\gamma)$ exists and is continuous on
$\Gamma$ for $\lambda$-almost every $x\in\mathcal X_0$.  Define
$M_r(x)=\sup_{\gamma\in\Gamma}
\|\nabla_\gamma^r a(x\mid\gamma)\|$ on $\mathcal X_0$ and set
$M_r(x)=0$ on $\mathcal X_0^c$.  The functions $M_r$ are measurable and
satisfy
$ P_0\Lambda_{r_{\max}}<\infty$ and
$Q\Lambda_{r_{\max}}<\infty$, where
$\Lambda_1=(1+M_0+M_1)^2$ and
$\Lambda_3=\Lambda_1+M_1^3+(1+M_1)M_2+M_3$.

\item[(A3)]
For the fixed value of $\nu$ under consideration:
\begin{enumerate}

\item[(i)]
Almost surely, there exists a measurable sequence
$\hat\gamma_n\in
\operatorname*{arg\,min}_{\gamma\in\Gamma}f_{n,\nu}(\gamma)$
for all sufficiently large $n$.

\item[(ii)]
The population loss is well separated at $\gamma_0$: for every
$\epsilon>0$,
$$
f_\nu(\gamma_0)
<
\inf_{\substack{\gamma\in\Gamma\\
\|\gamma-\gamma_0\|\ge\epsilon}}
f_\nu(\gamma),
$$
where the infimum of the empty set is understood to be $+\infty$.

\item[(iii)]
The matrix
$I_\nu=P_0
[\{1-\rho_{0, \nu}\}ss^\top]$ is nonsingular.

\end{enumerate}

\item[(A4)]
The prior is a probability distribution on $\Gamma$ with density $\pi$
with respect to Lebesgue measure.  The density $\pi$ is continuous at
$\gamma_0$ and satisfies $\pi(\gamma_0)>0$.
\end{enumerate}
\end{as}

\subsection{Preliminary results}

\begin{lem}
\label{lem:basic-ulln}
Under Assumptions~\textnormal{(A1)} and \textnormal{(A2)} with $r_{\max}\ge1$,
$$
S_n:=\sup_{\gamma\in\Gamma}
|f_{n,\nu}(\gamma)-f_\nu(\gamma)|
\overset{\text{a.s.}}\to0.
$$
Moreover, $f_\nu$ is continuous on $\Gamma$.
\end{lem}

\begin{proof}
For almost every $x$ and all $\gamma,\gamma'\in\Gamma$,
$|\phi(x\mid\gamma)|\le\log(1+\nu)+M_0(x)$ and
$|\psi(x\mid\gamma)|\le\log(1+\nu^{-1})+M_0(x)$.
These integrable envelope bounds ensure that the strong law of large numbers is applicable for each fixed $\gamma$.
Next, since $\Gamma$ is convex, the line segment joining
$\gamma$ and $\gamma'$ is contained in $\Gamma$.  The mean value
inequality and the derivative bounds yield
$|\phi(x\mid\gamma)-\phi(x\mid\gamma')|
\le M_1(x)\|\gamma-\gamma'\|$ and
$|\psi(x\mid\gamma)-\psi(x\mid\gamma')|
\le M_1(x)\|\gamma-\gamma'\|.$  These Lipschitz bounds are used below to upgrade pointwise convergence
to uniform convergence. 

For each fixed $\gamma\in\Gamma$, the strong law of large numbers yields
$f_{n,\nu}(\gamma)\overset{\text{a.s.}}\to f_\nu(\gamma)$.  Moreover,
$|f_{n,\nu}(\gamma)-f_{n,\nu}(\gamma')|
\le L_n\|\gamma-\gamma'\|,$ where
$L_n=P_nM_1+\nu Q_mM_1$ and
$L_n\overset{\text{a.s.}}\to L:=P_0M_1+\nu QM_1<\infty$.
Following the argument used in the proof of
\citet[Lemma~2]{matsubara2022robust}, the preceding random Lipschitz bound implies that $(f_{n,\nu})_{n\ge1}$ is strongly stochastically
equicontinuous on $\Gamma$ \citep[Theorem~21.10]{davidson1994stochastic}.  Therefore, by
\citet[Theorem~21.8]{davidson1994stochastic},
$
\sup_{\gamma\in\Gamma}
|f_{n,\nu}(\gamma)-f_\nu(\gamma)|
\overset{\text{a.s.}}\to0.
$

Finally, the corresponding deterministic bound
$
|f_\nu(\gamma)-f_\nu(\gamma')|
\le L\|\gamma-\gamma'\|
$
shows that $f_\nu$ is Lipschitz continuous on $\Gamma$.
\end{proof}

\begin{lem}
\label{lem:deriv-ident}
Suppose Assumptions \textnormal{(A1)} and \textnormal{(A2)} hold with $r_{\max}=3$.  Then
$f_\nu$ is twice continuously differentiable in a neighborhood of
$\gamma_0$, with
$
\nabla_\gamma f_\nu(\gamma_0)=0,~
\nabla_\gamma^2f_\nu(\gamma_0)=I_\nu.
$
Moreover, for every sufficiently small closed ball
$\overline B\subset\Gamma$ centered at $\gamma_0$,
\begin{equation}
\label{eq:hessian-ulln}
\sup_{\gamma\in\overline B}
\|\nabla_\gamma^2f_{n,\nu}(\gamma)
-\nabla_\gamma^2f_\nu(\gamma)\|
\overset{\text{a.s.}}\to0.
\end{equation}
\end{lem}

\begin{proof}
Write $\rho=\rho_\nu(x\mid\gamma)$, and let all derivatives below
be taken with respect to $\gamma$.  On $\mathcal X_0$,
$\nabla_\gamma\phi=-(1-\rho)\nabla_\gamma a$,
$\nabla_\gamma\psi=\rho\nabla_\gamma a$,
$\nabla_\gamma^2\phi=\rho(1-\rho)\nabla_\gamma a\nabla_\gamma a^\top
-(1-\rho)\nabla_\gamma^2a$, and
$\nabla_\gamma^2\psi=\rho(1-\rho)\nabla_\gamma a\nabla_\gamma a^\top
+\rho\nabla_\gamma^2a$.  On $\mathcal X_0^c$, $\psi(\cdot\mid\gamma)$
and all of its $\gamma$-derivatives are zero.  Since
$P_0(\mathcal X_0^c)=0$, the arbitrary finite extension of
$\phi(\cdot\mid\gamma)$ to $\mathcal X_0^c$ does not affect
$P_0\phi$ or $P_n\phi$ almost surely.

Fix a sufficiently small closed ball $\overline B$ centered at
$\gamma_0$.  By Assumption~\textnormal{(A1)}, there exists an open ball
$B_+$ such that $\overline B\subset B_+\subset\Gamma$.  For
$\ell\in\{\phi,\psi\}$, the preceding derivative formulas, the
third-order chain rule, and $0\le\rho\le1$ imply that there exists a
deterministic constant $C<\infty$, independent of $x$, $\gamma$, and
$n$, such that
$\sup_{\gamma\in B_+}\|\nabla_\gamma^r\ell(x\mid\gamma)\|
\le C\Lambda_3(x)$ for $r=1,2,3$ and $x\in\mathcal X_0$, where
$\Lambda_3$ is the envelope defined in
Assumption~\textnormal{(A2)}.  For $\ell=\psi$, the same bound holds
on $\mathcal X_0^c$ because its derivatives vanish there.

By Assumption~\textnormal{(A2)},
$P_0\Lambda_3+\nu Q\Lambda_3<\infty$.  Hence, for $r=1,2$,
differentiation under the integral sign gives
$\nabla_\gamma^r f_\nu(\gamma)
=P_0\nabla_\gamma^r\phi(\cdot\mid\gamma)
+\nu Q\nabla_\gamma^r\psi(\cdot\mid\gamma)$ for every
$\gamma\in B_+$.  Since the pointwise derivatives are continuous in
$\gamma$ and dominated by the same integrable envelope, the dominated
convergence theorem shows that these derivatives are continuous on
$B_+$.  Thus, $f_\nu\in C^2(B_+)$ and is therefore twice continuously
differentiable in a neighborhood of $\gamma_0$.

Moreover,
$\sup_{\gamma\in B_+}|\nabla_\gamma^3 f_{n,\nu}(\gamma)|
\le C{P_n\Lambda_3+\nu Q_m\Lambda_3}$ almost surely.  By the strong
law of large numbers,
$P_n\Lambda_3+\nu Q_m\Lambda_3
\overset{\text{a.s.}}\to
P_0\Lambda_3+\nu Q\Lambda_3<\infty$, and therefore\\
$\limsup_{n\to\infty}\sup_{\gamma\in B_+}
|\nabla_\gamma^3 f_{n,\nu}(\gamma)|<\infty$ almost surely.
For each fixed $\gamma\in B_+$, the strong law of large numbers,
applied componentwise, gives
$\nabla_\gamma^2 f_{n,\nu}(\gamma)
\overset{\text{a.s.}}\to
P_0\nabla_\gamma^2\phi(\cdot\mid\gamma)
+\nu Q\nabla_\gamma^2\psi(\cdot\mid\gamma)
=\nabla_\gamma^2 f_\nu(\gamma)$.
Moreover, by the mean value inequality, each component of
$\nabla_\gamma^2 f_{n,\nu}$ is Lipschitz on $\overline B$ with a random
Lipschitz constant bounded by
$\sup_{\gamma\in B_+}
|\nabla_\gamma^3 f_{n,\nu}(\gamma)|$, whose limit superior is finite
almost surely.  Thus, the Hessian components are
strongly stochastically equicontinuous \citep[Theorem~21.10]{davidson1994stochastic}.
Applying
\citet[Theorem~21.8]{davidson1994stochastic} componentwise and using
the equivalence of norms in finite dimensions yields
$
\sup_{\gamma\in\overline B}
\|\nabla_\gamma^2 f_{n,\nu}(\gamma)
-\nabla_\gamma^2 f_\nu(\gamma)\|
\overset{\text{a.s.}}\to0.
$

Correct specification gives
$\rho_{0, \nu}(x)=p_0(x)/{p_0(x)+\nu q(x)}$ for $\lambda$-almost every
$x\in\mathcal X_0$.  Since both $p_0$ and $\rho_{0,\nu}$ vanish on
$\mathcal X_0^c$, it follows that
\begin{equation}
\label{eq:S-density-id}
{1-\rho_{0, \nu}(x)}p_0(x)=\nu\rho_{0, \nu}(x)q(x)
\end{equation}
for $\lambda$-almost every $x$.
Multiplying \eqref{eq:S-density-id} by $s(x)=\nabla_\gamma a(x\mid\gamma_0)$ and integrating shows
that the two score terms cancel.  Multiplying it on $\mathcal X_0$ by
$\nabla_\gamma^2a(x\mid\gamma_0)$ and integrating shows that the
Hessian terms containing $\nabla_\gamma^2a(x\mid\gamma_0)$ cancel.
Finally, multiplying it by $(1-\rho_{0, \nu}(x))s(x)s(x)^\top$ and integrating
gives
$\nu Q[\rho_{0, \nu}(1-\rho_{0, \nu})ss^\top]
=P_0[(1-\rho_{0, \nu})^2ss^\top]$.  Therefore,
$\nabla_\gamma f_\nu(\gamma_0)
=-P_0[(1-\rho_{0, \nu})s]+\nu Q[\rho_{0, \nu}s]=0$ and
$\nabla_\gamma^2f_\nu(\gamma_0)
=P_0[\rho_{0, \nu}(1-\rho_{0, \nu})ss^\top]
+\nu Q[\rho_{0, \nu}(1-\rho_{0, \nu})ss^\top]
=P_0[(1-\rho_{0, \nu})ss^\top]=I_\nu$.
\end{proof}

\begin{lem}
\label{lem:Mclt}
Suppose Assumptions~\textnormal{(A1)}--\textnormal{(A3)} hold with
$r_{\max}=3$, and let $\hat\gamma_n$ be the measurable minimizer in
Assumption~\textnormal{(A3)(i)}.  Then
$\hat\gamma_n\overset{\text{a.s.}}\to\gamma_0$ and
$$
\sqrt n(\hat\gamma_n-\gamma_0)
\overset{d}\to
N(0,I_\nu^{-1}J_\nu I_\nu^{-1}),
$$
where
$\mu_\nu=P_0[(1-\rho_{0, \nu})s]$ and
$J_\nu=I_\nu-\{(1+\nu)/\nu\}\mu_\nu\mu_\nu^\top$.  Moreover,
$0\preceq J_\nu\preceq I_\nu$.
\end{lem}

\begin{proof}
Let $A_\epsilon=\{\gamma\in\Gamma:\|\gamma-\gamma_0\|\ge\epsilon\}$.
By Assumption~\textnormal{(A3)(ii)}, for every $\epsilon>0$,
$c_\epsilon=\inf_{\gamma\in A_\epsilon}
\{f_\nu(\gamma)-f_\nu(\gamma_0)\}>0$.
On the event $\{\hat\gamma_n\in A_\epsilon\}$, the definition of
$c_\epsilon$ gives
$f_\nu(\hat\gamma_n)\ge f_\nu(\gamma_0)+c_\epsilon$.
Since $\hat\gamma_n$ minimises $f_{n,\nu}$, we also have
$f_{n,\nu}(\hat\gamma_n)\le f_{n,\nu}(\gamma_0)$.  Recalling that
$S_n=\sup_{\gamma\in\Gamma}
|f_{n,\nu}(\gamma)-f_\nu(\gamma)|$, it follows that
$f_\nu(\gamma_0)+c_\epsilon
\le f_\nu(\hat\gamma_n)
\le f_{n,\nu}(\hat\gamma_n)+S_n
\le f_{n,\nu}(\gamma_0)+S_n
\le f_\nu(\gamma_0)+2S_n$.
Consequently,
$\{\|\hat\gamma_n-\gamma_0\|\ge\epsilon\}
\subset\{2S_n\ge c_\epsilon\}$.
By Lemma~\ref{lem:basic-ulln},
$S_n\overset{\text{a.s.}}\to0$.  Thus, almost surely,
$\|\hat\gamma_n-\gamma_0\|<\epsilon$ for all sufficiently large $n$.
Applying this argument to $\epsilon=1/k$, $k\in\mathbb N$, yields
$\hat\gamma_n\overset{\text{a.s.}}\to\gamma_0$.

Since $\Gamma$ is open and $\gamma_0\in\Gamma$, there exists
$\delta>0$ such that
$\overline B(\gamma_0,\delta)\subset\Gamma$.  By consistency,
$\hat\gamma_n\in B(\gamma_0,\delta)$ for all sufficiently large $n$,
almost surely.  Hence $\hat\gamma_n$ is eventually an interior minimizer
of $f_{n,\nu}$, and therefore
$\nabla_\gamma f_{n,\nu}(\hat\gamma_n)=0$ eventually, almost surely.
Define $g_1=-(1-\rho_{0, \nu})s$ and $g_0=\rho_{0, \nu}s$.  By
\eqref{eq:S-density-id},
$P_0[(1-\rho_{0, \nu})s]=\nu Q[\rho_{0, \nu}s]$.  Since
$\mu_\nu=P_0[(1-\rho_{0, \nu})s]$, we have
$P_0g_1=-\mu_\nu$ and $Qg_0=\mu_\nu/\nu$, and hence
$P_0g_1+\nu Qg_0=0$.  It follows that
\begin{equation}
\label{eq:score-decomposition}
\sqrt n\,\nabla_\gamma f_{n,\nu}(\gamma_0)
=
\sqrt n(P_n-P_0)g_1
+
\nu\sqrt n(Q_m-Q)g_0.
\end{equation}
The moment conditions in Assumption~\textnormal{(A2)} imply that
$P_0\|g_1\|^2<\infty$ and $Q\|g_0\|^2<\infty$.  Therefore, by the
multivariate central limit theorem,
$\sqrt n(P_n-P_0)g_1
\overset{d}\to
N(0,\operatorname{Var}_{p_0}(g_1))$.
Since $m=\nu n$, we have
$\nu\sqrt n(Q_m-Q)g_0
=\sqrt{\nu}\sqrt m(Q_m-Q)g_0$, and hence
$\nu\sqrt n(Q_m-Q)g_0
\overset{d}\to
N(0,\nu\operatorname{Var}_q(g_0))$.
The data and noise samples are independent, so the two terms on the
right-hand side of \eqref{eq:score-decomposition} converge jointly to
independent Gaussian random vectors.  Consequently,
\begin{equation}
\label{eq:gradient-clt}
\sqrt n\,\nabla_\gamma f_{n,\nu}(\gamma_0)
\overset{d}\to
N(0,J_\nu),
\qquad
J_\nu
=
\operatorname{Var}_{p_0}(g_1)
+
\nu\operatorname{Var}_q(g_0).
\end{equation}

We next calculate $J_\nu$.  Its raw second-moment part is
$P_0[g_1g_1^\top]+\nu Q[g_0g_0^\top]
=
P_0[(1-\rho_{0, \nu})^2ss^\top]
+\nu Q[\rho_{0, \nu}^2ss^\top]$.
Multiplying \eqref{eq:S-density-id} by $\rho_{0, \nu}ss^\top$ and integrating
gives
$\nu Q[\rho_{0, \nu}^2ss^\top]
=
P_0[\rho_{0, \nu}(1-\rho_{0, \nu})ss^\top]$.
Therefore,
$P_0[g_1g_1^\top]+\nu Q[g_0g_0^\top]
=
P_0[(1-\rho_{0, \nu})ss^\top]
=
I_\nu$.
Moreover,
$(P_0g_1)(P_0g_1)^\top=\mu_\nu\mu_\nu^\top$, while
$\nu(Qg_0)(Qg_0)^\top
=\nu(\mu_\nu/\nu)(\mu_\nu/\nu)^\top
=\nu^{-1}\mu_\nu\mu_\nu^\top$.
It follows that
\begin{equation}
\label{eq:Jnu-form}
J_\nu
=
I_\nu
-
\frac{1+\nu}{\nu}\mu_\nu\mu_\nu^\top.
\end{equation}
The variance representation in \eqref{eq:gradient-clt} shows that
$J_\nu\succeq0$.  On the other hand,
$I_\nu-J_\nu
=\{(1+\nu)/\nu\}\mu_\nu\mu_\nu^\top\succeq0$ by
\eqref{eq:Jnu-form}.  Hence $0\preceq J_\nu\preceq I_\nu$.

It remains to transfer the gradient limit to $\hat\gamma_n$.  Fix a
sufficiently small closed ball
$\overline B=\overline B(\gamma_0,\delta)\subset\Gamma$ for which
Lemma~\ref{lem:deriv-ident} applies, and set
$\Delta_n=\hat\gamma_n-\gamma_0$.  By consistency,
$\hat\gamma_n\in\overline B$ eventually, almost surely.  Since
$\overline B$ is convex, the line segment
$\{\gamma_0+t\Delta_n:0\le t\le1\}$ is then contained in
$\overline B$.
Applying the fundamental theorem of calculus to the map
$t\mapsto\nabla_\gamma f_{n,\nu}(\gamma_0+t\Delta_n)$ along the line
segment joining $\gamma_0$ and $\hat\gamma_n$, and using the
first-order condition
$\nabla_\gamma f_{n,\nu}(\hat\gamma_n)=0$, gives
\begin{equation}
\label{eq:taylor-score}
0
=
\nabla_\gamma f_{n,\nu}(\gamma_0)
+
\overline H_n\Delta_n,
\qquad
\overline H_n
=
\int_0^1
\nabla_\gamma^2f_{n,\nu}
(\gamma_0+t\Delta_n)\,\mathrm dt.
\end{equation}
We now show that $\overline H_n\overset{\text{a.s.}}\to I_\nu$.
Indeed,
\begin{equation}
\label{eq:Hbar-bound}
\begin{aligned}
\|\overline H_n-I_\nu\|
&\le
\sup_{\gamma\in\overline B}
\|
\nabla_\gamma^2f_{n,\nu}(\gamma)
-
\nabla_\gamma^2f_\nu(\gamma)
\| \\
&\quad+
\sup_{0\le t\le1}
\|
\nabla_\gamma^2f_\nu(\gamma_0+t\Delta_n)
-
\nabla_\gamma^2f_\nu(\gamma_0)
\|.
\end{aligned}
\end{equation}
The first term on the right-hand side of \eqref{eq:Hbar-bound}
converges almost surely to zero by \eqref{eq:hessian-ulln}.  For the
second term, $\Delta_n\overset{\text{a.s.}}\to0$, and
$\nabla_\gamma^2f_\nu$ is continuous at $\gamma_0$ by
Lemma~\ref{lem:deriv-ident}.  Hence that term also converges almost
surely to zero.  Since
$\nabla_\gamma^2f_\nu(\gamma_0)=I_\nu$, we obtain
$\overline H_n\overset{\text{a.s.}}\to I_\nu$.
By Assumption~\textnormal{(A3)}, $I_\nu$ is nonsingular.  Therefore,
$\overline H_n$ is nonsingular for all sufficiently large $n$, almost
surely, and
$\overline H_n^{-1}\overset{\text{a.s.}}\to I_\nu^{-1}$.
Rearranging \eqref{eq:taylor-score} gives
\begin{equation}
\label{eq:M-estimator-linearisation}
\sqrt n(\hat\gamma_n-\gamma_0)
=
-\overline H_n^{-1}
\sqrt n\,\nabla_\gamma f_{n,\nu}(\gamma_0).
\end{equation}
Combining \eqref{eq:gradient-clt} and
\eqref{eq:M-estimator-linearisation}, Slutsky's theorem yields
$
\sqrt n(\hat\gamma_n-\gamma_0)
\overset{d}\to
N\!\left(
0,
I_\nu^{-1}J_\nu I_\nu^{-1}
\right).
$
\end{proof}

\subsection{Proof of Theorem~\ref{thm:conc}}

\begin{proof}
Fix $\epsilon>0$ and let
$
A_\epsilon
=
\{\gamma\in\Gamma:
\|\gamma-\gamma_0\|\ge\epsilon\}.
$
If $A_\epsilon=\varnothing$, then
$\Pi_n(A_\epsilon)=0$ and there is nothing to prove.
Hence, suppose that $A_\epsilon\neq\varnothing$ and define
$
c_\epsilon
=
\inf_{\gamma\in A_\epsilon}
\{f_\nu(\gamma)-f_\nu(\gamma_0)\}
$
as in Lemma \ref{lem:Mclt}.
By Assumption~\textnormal{(A3)(ii)}, $c_\epsilon>0$.
By continuity of $f_\nu$ at $\gamma_0$, there exists an open
neighborhood $U_\epsilon\subset\Gamma$ of $\gamma_0$ such that
$
f_\nu(\gamma)
\le
f_\nu(\gamma_0)+{c_\epsilon}/{4},
~
\gamma\in U_\epsilon.
$
Since $\pi$ is continuous at $\gamma_0$ and
$\pi(\gamma_0)>0$, Assumption~\textnormal{(A4)} implies
$\Pi(U_\epsilon)>0$.

On the almost-sure event on which $S_n=\sup_{\gamma\in\Gamma}
|f_{n,\nu}(\gamma)-f_\nu(\gamma)|\to0$, uniform convergence gives
$
f_{n,\nu}(\gamma)
\ge
f_\nu(\gamma_0)+c_\epsilon-S_n,
\gamma\in A_\epsilon,
$
and
$
f_{n,\nu}(\gamma)
\le
f_\nu(\gamma_0)+{c_\epsilon}/{4}+S_n,~
\gamma\in U_\epsilon.
$
Consequently,
\[
\begin{split}
\Pi_n(A_\epsilon)
&=
\frac{
\int_{A_\epsilon}
\exp\{-nf_{n,\nu}(\gamma)\}\Pi(d\gamma)
}{
\int_{\Gamma}
\exp\{-nf_{n,\nu}(\gamma)\}\Pi(d\gamma)
}\\
&\le
\frac{
\exp[-n\{f_\nu(\gamma_0)+c_\epsilon-S_n\}]
\Pi(A_\epsilon)
}{
\exp[-n\{f_\nu(\gamma_0)+c_\epsilon/4+S_n\}]
\Pi(U_\epsilon)
}\\
&\le
\Pi(U_\epsilon)^{-1}
\exp\left[
-n\left\{\frac{3c_\epsilon}{4}-2S_n\right\}
\right].
\end{split}
\]
Since $S_n\overset{\text{a.s.}}\to0$ (Lemma \ref{lem:basic-ulln}), the right-hand side converges to zero
almost surely. 

The preceding argument establishes the stronger joint concentration result
$\Pi_n(\{\gamma\in\Gamma:\|\gamma-\gamma_0\|\ge\epsilon\})
\overset{\text{a.s.}}\to0$.
Let $\Pi_n^\theta=\Pi_n\circ R_\theta^{-1}$ denote the
$\theta$-marginal posterior, where
$R_\theta(\gamma)=\theta$.  Then
$\Pi_n^\theta(\{\theta\in\Theta:\|\theta-\theta_0\|\ge\epsilon\})
=
\Pi_n(\{\gamma\in\Gamma:\|\theta-\theta_0\|\ge\epsilon\})$.
Since $\gamma=(\theta^\top,\beta)^\top$ and the Euclidean norm satisfies
$\|\theta-\theta_0\|\le\|\gamma-\gamma_0\|$, we have
$\{\gamma\in\Gamma:\|\theta-\theta_0\|\ge\epsilon\}
\subset
\{\gamma\in\Gamma:\|\gamma-\gamma_0\|\ge\epsilon\}$.
Consequently,
$\Pi_n^\theta(\{\theta\in\Theta:\|\theta-\theta_0\|\ge\epsilon\})
\le
\Pi_n(\{\gamma\in\Gamma:\|\gamma-\gamma_0\|\ge\epsilon\})
\overset{\text{a.s.}}\to0$.
Equivalently,
$\Pi_n^\theta(\{\theta\in\Theta:\|\theta-\theta_0\|<\epsilon\})
\overset{\text{a.s.}}\to1$,
which proves the stated concentration result for the
$\theta$-marginal posterior.
\end{proof}

\subsection{Proof of Theorem~\ref{thm:bvm}}

\begin{proof}
Our aim is to verify the conditions of
\citet[Theorem~4]{miller2021general}.  Let $\Omega_0$ be an event of
probability one on which the almost-sure conclusions of
Lemmas~\ref{lem:basic-ulln}, \ref{lem:deriv-ident}, and
\ref{lem:Mclt} hold.  Thus, on $\Omega_0$,
$S_n:=\sup_{\gamma\in\Gamma}
|f_{n,\nu}(\gamma)-f_\nu(\gamma)|\to0$,
$\hat\gamma_n\to\gamma_0$, the Hessians converge locally uniformly,
the third derivatives are locally eventually bounded, and
$\nabla_\gamma f_{n,\nu}(\hat\gamma_n)=0$ for all sufficiently large
$n$.  Fix a sample point in $\Omega_0$.  Extend the
prior density to $\mathbb R^{p+1}$ by
$\overline\pi(\gamma)=\pi(\gamma)\mathbf 1_\Gamma(\gamma)$, and extend
the empirical loss by setting
$\overline f_{n,\nu}(\gamma)=f_{n,\nu}(\gamma)$ for
$\gamma\in\Gamma$ and
$\overline f_{n,\nu}(\gamma)=f_{n,\nu}(\hat\gamma_n)+1$ for
$\gamma\notin\Gamma$.  Since
$\gamma_0\in\operatorname{int}(\Gamma)$,
$\overline\pi$ is continuous at $\gamma_0$ and
$\overline\pi(\gamma_0)>0$, while the extensions do not alter the
posterior distribution or its normalizing constant.

We first verify the local quadratic approximation required by
\citet[Theorem~4]{miller2021general}.  Define
$H_n=\nabla_\gamma^2 f_{n,\nu}(\hat\gamma_n).$
Choose $\delta>0$ sufficiently small that
$\overline B(\gamma_0,2\delta)\subset\Gamma$ and the local uniform
Hessian convergence in Lemma~\ref{lem:deriv-ident} holds on this
closed ball. Since $\hat\gamma_n\to\gamma_0$, we have
$
\|\hat\gamma_n-\gamma_0\|<\delta
$
for all sufficiently large $n$. Consequently,
$
\hat\gamma_n\in B(\gamma_0,\delta)
\subset\overline B(\gamma_0,2\delta),
\overline B(\hat\gamma_n,\delta)
\subset B(\gamma_0,2\delta)
\subset\Gamma
$
for all sufficiently large $n$.  
Using $I_\nu=\nabla_\gamma^2f_\nu(\gamma_0)$ and the
triangle inequality, $\|H_n-I_\nu\|$ is bounded above by
$\sup_{\gamma\in\overline B(\gamma_0,2\delta)}
\|\nabla_\gamma^2f_{n,\nu}(\gamma)
-\nabla_\gamma^2f_\nu(\gamma)\|
+\|\nabla_\gamma^2f_\nu(\hat\gamma_n)
-\nabla_\gamma^2f_\nu(\gamma_0)\|$.  The first term converges to zero
almost surely by Lemma~\ref{lem:deriv-ident}, while the second
converges to zero by the continuity of
$\nabla_\gamma^2f_\nu$ and the consistency of $\hat\gamma_n$.
Consequently, $H_n\overset{\text{a.s.}}\to I_\nu$.  
Each $H_n$ is symmetric because $f_{n,\nu}$ is twice continuously
differentiable in a neighborhood of $\hat\gamma_n$.  Moreover, for
every $v\in\mathbb R^{p+1}$,
$v^\top I_\nu v
=P_0[\{1-\rho_{0, \nu}(X)\}\{v^\top s(X)\}^2]\ge0$.
Thus $I_\nu$ is symmetric positive semidefinite.  Since
$I_\nu$ is nonsingular by Assumption~\textnormal{(A3)(iii)}, it has no
zero eigenvalues and is therefore positive definite.

For every $h\in\mathbb R^{p+1}$, define
$r_n(h)=\overline f_{n,\nu}(\hat\gamma_n+h)
-\overline f_{n,\nu}(\hat\gamma_n)-(1/2)h^\top H_nh$.
Then the required quadratic representation holds identically on
$\mathbb R^{p+1}$.  If $\|h\|\le\delta$, the line segment joining
$\hat\gamma_n$ and $\hat\gamma_n+h$ is contained in
$\overline B(\hat\gamma_n,\delta)\subset
\overline B(\gamma_0,2\delta)\subset\Gamma$ for all sufficiently large
$n$.  Since
$\nabla_\gamma f_{n,\nu}(\hat\gamma_n)=0$, Taylor's theorem to second
order with integral remainder shows that $r_n(h)$ is the corresponding
third-order remainder.  By the local uniform third-derivative bound in
Lemma~\ref{lem:deriv-ident}, there exists a finite
sample-path-dependent constant $C_0$ such that
$|r_n(h)|\le C_0\|h\|^3$ for all $\|h\|\le\delta$ and all sufficiently
large $n$. 

It remains to verify separation.  Fix $\epsilon>0$ and let
$A_{\epsilon/2}
=\{\gamma\in\Gamma:\|\gamma-\gamma_0\|\ge\epsilon/2\}$.  Suppose first
that $A_{\epsilon/2}$ is nonempty.  By
Assumption (A3) (ii),
$c_{\epsilon/2}:=\inf_{\gamma\in A_{\epsilon/2}}
\{f_\nu(\gamma)-f_\nu(\gamma_0)\}>0$.  Since
$\|\hat\gamma_n-\gamma_0\|<\epsilon/2$ eventually,
$\gamma\in\Gamma$ and
$\|\gamma-\hat\gamma_n\|\ge\epsilon$ imply
$\gamma\in A_{\epsilon/2}$.  
Moreover, since $\hat\gamma_n$ minimises $f_{n,\nu}$,
$f_{n,\nu}(\hat\gamma_n)\le f_{n,\nu}(\gamma_0)$.  Therefore, for
every $\gamma\in\Gamma$ satisfying
$\|\gamma-\hat\gamma_n\|\ge\epsilon$,
$f_{n,\nu}(\gamma)-f_{n,\nu}(\hat\gamma_n)
\ge f_\nu(\gamma)-f_\nu(\gamma_0)-2S_n
\ge c_{\epsilon/2}-2S_n$.
Outside $\Gamma$, the corresponding loss
difference equals $1$.  Therefore,
$\liminf_{n\to\infty}
\inf_{\|\gamma-\hat\gamma_n\|\ge\epsilon}
\{\overline f_{n,\nu}(\gamma)
-\overline f_{n,\nu}(\hat\gamma_n)\}
\ge\min\{c_{\epsilon/2},1\}>0$.  If $A_{\epsilon/2}$ is empty, then
every $\gamma\in\Gamma$ satisfies
$\|\gamma-\gamma_0\|<\epsilon/2$, and
$\|\hat\gamma_n-\gamma_0\|<\epsilon/2$ eventually; thus
$\|\gamma-\hat\gamma_n\|<\epsilon$ throughout $\Gamma$, so the same
infimum is taken only outside $\Gamma$ and equals $1$.

Since $f_{n,\nu}$ is finite and nonnegative on $\Gamma$ and
the prior is a probability measure, the posterior normalizing
constant is finite and strictly positive.  All conditions of
\citet[Theorem~4]{miller2021general} are therefore satisfied, and the
density of $\sqrt n(\gamma-\hat\gamma_n)$ under the NC-Bayes posterior
converges in total variation to the density of
$N(0,I_\nu^{-1})$.

We have thus established the Bernstein--von Mises limit for the joint parameter $\gamma$.  It remains to transfer this result to the parameter of interest $\theta=R_\theta\gamma$ and to show that the limiting covariance of the resulting marginal posterior agrees with the asymptotic covariance of the minimum NC-loss estimator $\hat\theta_n$.
Since total variation distance cannot increase under a measurable
mapping, applying the projection $u\mapsto R_\theta u$ to the joint
Bernstein--von Mises limit yields
\[
\left\|
\Pi\left(
\sqrt n(\theta-\hat\theta_n)\in\cdot
\mid X_n,X_m^\ast
\right)
-
N\left(0,R_\theta I_\nu^{-1}R_\theta^\top\right)(\cdot)
\right\|_{\mathrm{TV}}
\to0.
\]
Thus, the first assertion follows with
$V_{\nu,\theta}=R_\theta I_\nu^{-1}R_\theta^\top$.
For the second assertion, let
$e_\beta=(0,\ldots,0,1)^\top\in\mathbb R^{p+1}$.  Since
$\partial a(x\mid\gamma)/\partial\beta=1$, we have
$s(x)^\top e_\beta=1$ for every $x\in\mathcal X_0$.  Therefore,
$
I_\nu e_\beta
=
P_0
\left[
\{1-\rho_{0, \nu}\}ss^\top e_\beta
\right]
=
P_0
\left[
\{1-\rho_{0, \nu}\}s
\right]
=
\mu_\nu.
$
Because $I_\nu$ is nonsingular,
$I_\nu^{-1}\mu_\nu=e_\beta$.  Recalling that
$
J_\nu
=
I_\nu
-
(1+\nu)\mu_\nu\mu_\nu^\top/\nu,
$
by \eqref{eq:Jnu-form},
we obtain
$
I_\nu^{-1}J_\nu I_\nu^{-1}
=
I_\nu^{-1}
-
({1+\nu})e_\beta e_\beta^\top/\nu.
$   
Since $R_\theta e_\beta=0$, it follows that
$
R_\theta I_\nu^{-1}J_\nu I_\nu^{-1}R_\theta^\top
=
R_\theta I_\nu^{-1}R_\theta^\top
=
V_{\nu,\theta}.
$
Applying the continuous mapping theorem to Lemma~\ref{lem:Mclt} therefore gives
$
\sqrt n(\hat\theta_n-\theta_0)
=
R_\theta\sqrt n(\hat\gamma_n-\gamma_0)
\overset{d}\to N(0,V_{\nu,\theta}).
$
\end{proof}

\subsection{Proof of Corollary~\ref{cor:large-noise-theta-calibration}}

\begin{proof}
For every sufficiently large fixed
$\nu\in\mathbb Q_{>0}$, Assumptions~\textnormal{(A1)}--\textnormal{(A4)}
with $r_{\max}=3$ allow us to apply
Theorem~\ref{thm:bvm}.  Hence, the covariance matrix in the Gaussian
Bernstein--von Mises limit for the $\theta$-marginal NC-Bayes posterior is
$
V_{\nu,\theta}
=
R_\theta I_\nu^{-1}R_\theta^\top.
$
It therefore remains to determine the limit of this matrix as
$\nu\to\infty$.
Write
$\rho_{0, \nu}(x)=\rho_\nu(x\mid\gamma_0)$.  By correct specification,
$$
\rho_{0, \nu}(x)
=
\frac{p_0(x)}{p_0(x)+\nu q(x)}
$$
on $\mathcal X_0$.  Since $q(x)>0$ for $P_0$-almost every $x$, we have
$1-\rho_{0, \nu}(x)\to1$ for $P_0$-almost every $x$ as
$\nu\to\infty$.
For every pair of indices $j,k$,
$
|\{1-\rho_{0, \nu}(x)\}s_j(x)s_k(x)|
\le
\|s(x)\|^2
\le
M_1(x)^2.
$
Assumption~\textnormal{(A2)} implies
$\mathbb E_{p_0}[M_1(X)^2]<\infty$, since $M_1^2\le\Lambda_1$.  The dominated convergence
theorem applied componentwise therefore yields
$
I_\nu=\mathbb E_{p_0}
[\{1-\rho_{0, \nu}(X)\}s(X)s(X)^\top]\to I_\infty,
$
where
$
I_\infty
=
\mathbb E_{p_0}
\left[
s(X)s(X)^\top
\right].
$
Let
$
t(X)=\nabla_\theta\log\tilde p(X\mid\theta_0).
$
Since
$
s(X)=(t(X)^\top,1)^\top,
$
we have
$$
I_\infty
=
\begin{pmatrix}
A & b\\
b^\top & 1
\end{pmatrix},
$$
where
$
A=\mathbb E_{p_0}[t(X)t(X)^\top]
$
and
$
b=\mathbb E_{p_0}[t(X)].
$
Moreover,
$
\nabla_\theta\log p(X\mid\theta_0)
=
t(X)-\nabla_\theta\log Z(\theta_0).
$
Since $\nabla_\theta\log Z(\theta_0)$ is nonrandom, covariance is
invariant under this translation, and hence
$$
\mathcal I(\theta_0)
=
\operatorname{Var}_{p_0}
\left[
\nabla_\theta\log p(X\mid\theta_0)
\right]
=
\operatorname{Var}_{p_0}[t(X)]
=
A-bb^\top.
$$
Thus, $\mathcal I(\theta_0)$ is the Schur complement of the lower-right
block of $I_\infty$.
Since $\mathcal I(\theta_0)$ is positive definite and the lower-right
block of $I_\infty$ equals $1$, the Schur complement criterion implies
that $I_\infty$ is positive definite.  Since
$I_\nu\to I_\infty$, it follows that $I_\nu$ is nonsingular for all
sufficiently large $\nu$, and continuity of matrix inversion on the set
of nonsingular matrices gives
$
I_\nu^{-1}\to I_\infty^{-1}.
$
Finally, the block inverse formula gives
$
R_\theta I_\infty^{-1}R_\theta^\top
=
(A-bb^\top)^{-1}
=
\mathcal I(\theta_0)^{-1}.
$
Consequently,
$$
V_{\nu,\theta}
=
R_\theta I_\nu^{-1}R_\theta^\top
\longrightarrow
R_\theta I_\infty^{-1}R_\theta^\top
=
\mathcal I(\theta_0)^{-1}
$$
$\text{as }\nu\to\infty$ through $\nu\in Q_{>0}$.
\end{proof}

\section{Details of posterior computation algorithms}
\label{sec:s2}
\subsection{Time-varying Density estimation}\label{sec:comp-nc-bayes-density}

Using data augmentation, the posterior samples can be generated via a simple Gibbs sampler, as described in the following step-by-step algorithm. 

\begin{enumerate}
\item
(Sampling from PG latent variable)\  
For $t=1,\ldots,T$ and $i=1,\ldots,n_t+m_t$, generate $\omega_{ti}$ from ${\rm PG}(1, \Phi_{ti}^\top\theta_t +\beta_t + C_{ti})$. 

\item
(Sampling from $\theta_t$) \ For $t=1,\ldots,T$, generate $\theta_t$ from $N(\widetilde{B}_{\theta_t}\widetilde{A}_{\theta_t}, \widetilde{B}_{\theta_t})$, where 
\begin{align*}
&\widetilde{B}_{\theta_t}=\left( b_{0t} I_L+ \sum_{i=1}^{n_t+m_t}\omega_{ti} \Phi_{ti}\Phi_{ti}^\top\right)^{-1}, \\ 
&\widetilde{A}_{\theta_t}=a_{0t}+\sum_{i=1}^{n_t+m_t}\Big\{s_{ti}-\frac{1}{2}-\omega_{ti}(\beta_t+C_{ti})\Big\} \Phi_{ti}
\end{align*}
with $b_{0t}=2\lambda^{-1}$ and $a_{0t}=\lambda^{-1}(\theta_{t-1}+\theta_{t+1})$ for $t=1,\ldots,T-1$ and $b_{0T}=\lambda^{-1}$ and $a_{0T}=\lambda^{-1}\theta_{T-1}$.

\item
(Sampling from $\beta_t$) \ For $t=1,\ldots,T$, generate $\beta_t$ from $N(\widetilde{B}_{\beta_t}\widetilde{A}_{\beta_t}, \widetilde{B}_{\beta_t})$, where 
$$
\widetilde{B}_{\beta_t}=\left( b_0^{-1}+ \sum_{i=1}^{n_t+m_t}\omega_{ti} \right)^{-1},  \ \ \ \ 
\widetilde{A}_{\beta_t}=\sum_{i=1}^{n_t+m_t}\Big\{s_{ti}-\frac{1}{2}-\omega_{ti}(\Phi_{ti}^\top \theta_t+C_{ti})\Big\}.
$$

\item
(Sampling from $\lambda$) \ Generate $\lambda$ from ${\rm IG}(n_1, \nu_1)$, where $n_1=n_0+TL/2$ and $\nu_1=\nu_0+\sum_{t=1}^T\sum_{l=1}^L (\theta_{tl}-\theta_{t-1,l})^2/2$.
\end{enumerate}

\subsection{Sparse Torus Graph}\label{sec:comp-nc-bayes-torus}
The posterior distribution from (\ref{eq:pos-noise}) is sampled using a Gibbs sampler that iterates through four main steps. The complete Gibbs sampler is described as follows:

\begin{enumerate}
\item (Sampling of $\omega_i$)\ For $i=1,\ldots,n+m$, sample the P\'olya--Gamma latent variables from:
$$
\omega_i \sim {\rm PG}(1,z(x_i)^\top\gamma+C(x_i)).
$$

\item (Sampling of $\gamma$)\ Conditional on the latent variables $\omega_i$ and the prior hyperparameters, sample $\gamma$ from its Gaussian full conditional $N(\tilde{b}_{\gamma}, \tilde{B}_{\gamma})$, where
\begin{align*}
\tilde{B}_{\gamma} &= \left(B_{\gamma 0}^{-1}+\sum_{i=1}^{n+m}\omega_i z(x_i)z(x_i)^\top\right)^{-1}, \\
\tilde{b}_{\gamma} &= \tilde{B}_{\gamma} \left\{ \sum_{i=1}^{n+m}\Big(s_i-\frac{1}{2}-\omega_iC(x_i)\Big) z(x_i) + B_{\gamma 0}^{-1}b_{\gamma 0} \right\}.
\end{align*}
Here, the prior covariance matrix $B_{\gamma 0}$ is determined by the shrinkage prior structure detailed in the next step, and we set the prior mean to zero, $b_{\gamma 0}=0$, so that $B_{\gamma 0}^{-1} b_{\gamma 0}$ vanishes in the mean update. Also, under a regularized horseshoe prior, we add $c^{-2}$ to the diagonal of the prior precision matrix $B_{\gamma 0}^{-1}$, where $c$ is the finite slab width; in our experiments, we set $c=1$.

\item (Sampling of Hyperparameters)\ Depending on the chosen prior, update the shrinkage hyperparameters.

\begin{itemize}
\item[ (a)] For the standard horseshoe prior:
Let $\phi_k$ be an element of $\phi$ and let $\lambda_k$ be its local shrinkage parameter for $k = 1,\dots, 2d^2$.
\begin{itemize}
    \item[-] Sample the local shrinkage components:
    \begin{align*}
    \lambda_k^2 \mid \phi_k, \tau^2, \nu_k &\sim \text{Inv-Gamma}\left(1, \frac{\phi_k^2}{2\tau^2} + \frac{1}{\nu_k}\right), \\
    \nu_k \mid \lambda_k^2 &\sim \text{Inv-Gamma}\left(1, 1 + \frac{1}{\lambda_k^2}\right).
    \end{align*}
    \item[-] Sample the global shrinkage components:
    \begin{align*}
    \tau^2 \mid \{\phi_k, \lambda_k^2\}_k, \xi &\sim \text{Inv-Gamma}\left(\frac{2d^2+1}{2}, \frac{1}{2}\sum_{k=1}^{2d^2} \frac{\phi_k^2}{\lambda_k^2} + \frac{1}{\xi}\right), \\
    \xi \mid \tau^2 &\sim \text{Inv-Gamma}\left(1, 1 + \frac{1}{\tau^2}\right).
    \end{align*}
\end{itemize}

\item[ (b)] For the grouped horseshoe prior on interaction parameters:
The main effects $\phi_{j(l)}$ are updated as in (a). For interaction parameters $\phi_{jk(l')}$:
\begin{itemize}
\item[-] Sample the group-level shrinkage components for each edge $(j,k)$:
\begin{align*}
u_{jk}^2 \mid - &\sim \text{Inv-Gamma}\left(\frac{5}{2}, \frac{1}{2}\sum_{l'=1}^4\frac{\phi_{jk(l')}^2}{\tau^2} + \frac{1}{t_{jk}}\right), \\
t_{jk} \mid u_{jk}^2 &\sim \text{Inv-Gamma}\left(1, 1 + \frac{1}{u_{jk}^2}\right).
\end{align*}
 \item[-] Sample the global shrinkage components:
\begin{align*}
\tau^2 \mid - &\sim \text{Inv-Gamma}\left(\frac{2d^2+1}{2}, \frac{1}{2}\left( \sum_{j,l} \frac{\phi_{j(l)}^2}{\lambda_{j(l)}^2} + \sum_{j<k,l'} \frac{\phi_{jk(l')}^2}{u_{jk}^2} \right) + \frac{1}{\xi}\right), \\
\xi \mid \tau^2 &\sim \text{Inv-Gamma}\left(1, 1 + \frac{1}{\tau^2}\right).
\end{align*}
\end{itemize}
\end{itemize}
\end{enumerate}

\subsection{Alternative: Generalized Bayesian Inference with Hyv\"arinen Score}\label{sec:comp-h-bayes}
As an alternative to NC-Bayes, we also implement a generalized Bayesian approach \citep{bissiri2016general}. This method utilizes the Hyv\"arinen score, whose specific form for the torus graph model is derived in \cite{klein2020torus}. For exponential families, the Hyv\"arinen-score loss can be expressed as
a quadratic function of the natural parameter $\phi$, yielding Gaussian
conjugacy under a Gaussian prior
\citep[see, e.g., Proposition~3.1 of][]{altamirano2023score}. Let the empirical loss function based on the observed data $X_n=\{x_1, \dots, x_n\}$ be
$$
L(\phi) = \frac{1}{2}\phi^\top \widehat{\Gamma} \phi - \phi^\top \widehat{H},
$$
where the matrix $\widehat{\Gamma}$ and vector $\widehat{H}$ are computed from the data as detailed in \cite{klein2020torus}.
The generalized posterior distribution for $\phi$ is then defined as:
$$
\pi_w(\phi\mid X_n) \propto \pi(\phi)\exp\{-nwL(\phi)\},
$$
where $\pi(\phi)$ is the prior distribution and $w > 0$ is a loss scaling parameter that controls the influence of the data on the posterior. Conditional on the horseshoe hyperparameters, the prior for $\phi$ is Gaussian. Combined with the quadratic form of $L(\phi)$, this yields a conditionally conjugate Gaussian full conditional distribution for $\phi$.

The MCMC algorithm for this method involves the following steps:
\begin{itemize}
\item[\bf 1.] \textbf{(Pre-computation)} Given the data $X_n$, first compute the matrices $\widehat{\Gamma}$ and $\widehat{H}$ based on the derivatives of the sufficient statistics, following the formulation in \cite{klein2020torus}. This is done once before the MCMC sampling.

\item[\bf 2.] \textbf{(MCMC Iteration)} At each iteration, sample the parameters and hyperparameters.
\begin{itemize}
\item[\bf (a)] \textbf{(Sampling of $\phi$)} Let the prior for $\phi$ be $N(0, B_{\phi 0})$, where $B_{\phi 0}$ is determined by the horseshoe shrinkage structure. Sample $\phi$ from its Gaussian full conditional posterior $N(\tilde{b}_{\phi}, \tilde{B}_{\phi})$, where
$
\tilde{B}_{\phi}^{-1} = B_{\phi 0}^{-1} + nw\widehat{\Gamma}, ~
\tilde{b}_{\phi} = \tilde{B}_{\phi} (nw\widehat{H}).
$

\item[\bf (b)] \textbf{(Sampling of Hyperparameters)} Sample the hyperparameters for the standard or grouped horseshoe prior. The full conditional distributions and update steps for these parameters are identical to those described in Step 3(a) and 3(b) of the NC-Bayes-based Gibbs sampler.
\end{itemize}
\end{itemize}

\section{Additional Simulation Results on Time-Varying Density Estimation }
\label{sec:sim-density-add}

Here, we provide the true time-varying densities and their point estimates by NC-Bayes and time-wise kernel density estimation.
The results are given in Figure~\ref{fig:sim-oneshot-supp}.
\begin{figure}[b!]
\centering
\includegraphics[width=\linewidth]{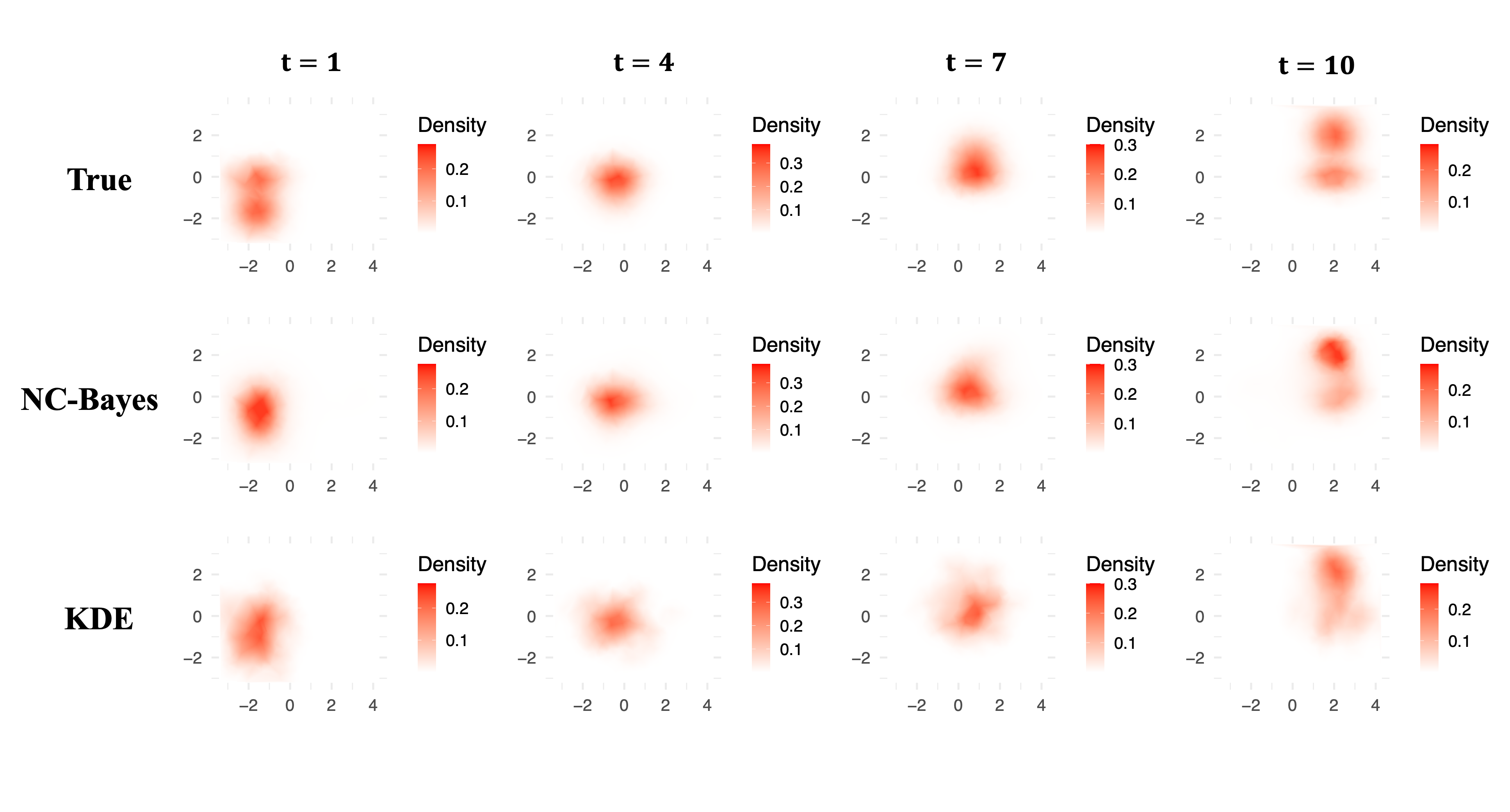}
\caption{ True time-varying density (upper), posterior mean of the time-varying density model fitted by NC-Bayes (middle) and time-wise KDE (lower), for four selected time points under Scenario 1. } 
\label{fig:sim-oneshot-supp}
\end{figure}

\section{Additional Simulation Results on Sparse Torus Graph}
\label{sec:simulation-torus-graph}

In addition to the simulation study presented in the main text, we conducted two further experiments with settings based on \citep{sukeda2025Directional} to assess the performance of NC-Bayes under different graph structures and dimensionalities.

\subsection{Scenario 1: 5-Node Cycle Graph}
This scenario was designed to evaluate the methods' performance on a small, well-defined sparse graph structure.

\begin{itemize}
    \item \textbf{Graph Structure}: We used a fixed graph with $d=5$ nodes. The true graph is a 5-cycle (a pentagonal structure). The specific edges defined in the simulation code are (1,3), (1,4), (2,4), (2,5), and (3,5).
    
    \item \textbf{True Parameter Values}: The main effect parameters were set to zero, i.e., $\phi_j = (0, 0)^\top$ for all $j=1,\ldots,5$. For the 5 pairs of nodes $(j,k)$ connected by an edge, the interaction parameters were set to $\phi_{jk}=(0.3, 0.3, 0.3, 0.3)^\top$. All other interaction parameters for non-connected pairs were set to zero.
    
    \item \textbf{Data Generation Algorithm}: A total of $n=1,000$ samples were generated from the true torus graph model using rejection sampling. Specifically, we first computed a uniform upper bound $U_{\max}$ for the unnormalized log-density over the entire domain. We then repeatedly drew candidate samples $x$ from a uniform distribution on $[0, 2\pi)^5$ and accepted each candidate with probability $\exp(U(x) - U_{\max})$, where $U(x)$ is the unnormalized log-density at $x$. This process was continued until $n=1,000$ samples were accepted.
\end{itemize}

The results for this scenario are presented in Tables~\ref{tab:results_star_Median} and~\ref{tab:results_star_CI}.
Under the median-based rule, NC-Bayes achieves perfect edge recovery when $\nu=1$, with recall, precision, and accuracy all equal to one, both with and without adaptive noise updating.
Under the CI-based rule, NC-Bayes also performs nearly perfectly for $\nu=1$, achieving perfect edge recovery together with high coverage probability regardless of whether the noise distribution is updated.
When the number of noise samples is increased to $\nu=5$, this strong performance is maintained under both detection rules; under the CI-based rule, precision and accuracy remain close to one, although the coverage probability decreases slightly.
In contrast, the performance of $\mathcal{H}$-Bayes remains sensitive to the loss-scaling parameter $w$: while $w=0.2$ yields essentially perfect performance, larger values of $w$ lead to increasingly more false-positive detections, resulting in lower precision and accuracy.
This deterioration is particularly evident under the CI-based rule, where the coverage probability also decreases substantially as $w$ increases.
Overall, these results show that NC-Bayes provides stable edge-selection and uncertainty-quantification performance across different numbers of noise samples and noise-updating schemes, whereas $\mathcal{H}$-Bayes remains strongly affected by the choice of $w$, consistent with the results reported in the main simulation study.

\begin{table}[b]
\centering
\caption{
Performance metrics for 5-dimensional torus graphs under the median-based edge detection rule,
averaged over 100 simulations.
An edge is detected when the absolute posterior median of any associated coefficient exceeds 0.100.
}
\vspace{2mm}
\begin{tabular}{llllrrr}
\hline
Method & $\nu$ & $w$ & Noise Update & Recall & Precision & Accuracy \\
\hline\hline
\multirow{4}{*}{NC-Bayes}
& \multirow{2}{*}{1} & -- & False & 1.000 & 1.000 & 1.000 \\
&                        & -- & True  & 1.000 & 1.000 & 1.000 \\
\cline{2-7}
& \multirow{2}{*}{5} & -- & False & 1.000 & 1.000 & 1.000 \\
&                        & -- & True  & 1.000 & 1.000 & 1.000 \\
\hline
\multirow{3}{*}{$\mathcal{H}$-Bayes}
& -- & 0.2 & -- & 1.000 & 1.000 & 1.000 \\
& -- & 0.5 & -- & 1.000 & 0.988 & 0.994 \\
& -- & 1.0 & -- & 1.000 & 0.926 & 0.960 \\
\hline
\end{tabular}
\label{tab:results_star_Median}
\end{table}

\begin{table}[t]
\centering
\caption{
Performance metrics for 5-dimensional torus graphs under the 90\% CI-based edge detection rule,
averaged over 100 simulations.
An edge is detected when its 90\% credible interval does not contain zero.
CP denotes the coverage probability for $\phi$.
}
\vspace{2mm}
\begin{tabular}{llllrrrr}
\hline
Method & $\nu$ & $w$ & Noise Update & CP (\%) & Recall & Precision & Accuracy \\
\hline\hline
\multirow{4}{*}{NC-Bayes}
& \multirow{2}{*}{1} & -- & False & 99.7 & 1.000 & 1.000 & 1.000 \\
&                        & -- & True  & 99.1 & 1.000 & 1.000 & 1.000 \\
\cline{2-8}
& \multirow{2}{*}{5} & -- & False & 98.0 & 1.000 & 0.998 & 0.999 \\
&                        & -- & True  & 97.0 & 1.000 & 0.998 & 0.999 \\
\hline
\multirow{3}{*}{$\mathcal{H}$-Bayes}
& -- & 0.2 & -- & 99.1 & 1.000 & 1.000 & 1.000 \\
& -- & 0.5 & -- & 93.3 & 1.000 & 0.935 & 0.965 \\
& -- & 1.0 & -- & 80.5 & 1.000 & 0.744 & 0.828 \\
\hline
\end{tabular}
\label{tab:results_star_CI}
\end{table}

\subsection{Scenario 2: 30-Dimensional Erd\"os-R\'enyi Graph}
This scenario was designed to assess performance in a higher-dimensional setting with a more complex, random structure.
\begin{itemize}
    \item \textbf{Graph Structure}: We used a fixed instance of an Erd\"os-R\'enyi random graph $G(d,p_e)$ with $d=30$ nodes and an edge probability of $p_e=0.1$. The graph was generated once using a fixed seed and this same structure was used for all 100 Monte Carlo simulations.
    
    \item \textbf{True Parameter Values}: All main effect parameters $\phi_j$ were set to zero. For pairs of nodes $(j,k)$ connected by an edge in the generated graph, the interaction parameters were set to $\phi_{jk}=(0.3, 0.3, 0, 0)^\top$, corresponding to only rotational dependence. All other interaction parameters were zero.
    
    \item \textbf{Data Generation Algorithm}: A total of $n=1,000$ samples were generated from the true model using a Gibbs sampler. Starting from a random initial state $x^{(0)} \in [0, 2\pi)^{30}$, we iteratively sampled each node's value $x_k$ from its full conditional distribution, which is a von-Mises distribution whose parameters depend on the current values of its neighbors in the graph. We ran the Gibbs sampler for 3,000 iterations, discarded the first 2,000 iterations as burn-in, and collected the subsequent 1,000 samples to form the dataset.
\end{itemize}

Tables~\ref{tab:results_ER_Median} and~\ref{tab:results_ER_CI} report the results under the median-based and CI-based rules, respectively.
For NC-Bayes, the performance depends on both $\nu$ and adaptive noise updating.
When $\nu=1$, edge recovery is highly accurate under the median-based rule, while under the CI-based rule adaptive noise updating improves recall from $0.776$ to $0.841$ while maintaining perfect precision and high coverage probability.
Increasing the number of noise samples to $\nu=5$ substantially improves the CI-based performance without noise updating, yielding nearly perfect recall, precision, and accuracy.
This behavior is consistent with Corollary~\ref{cor:large-noise-theta-calibration}, which shows that the asymptotic variance approaches that determined by the Fisher information as the amount of noise data increases.
However, for $\nu=5$, adaptive noise updating substantially deteriorates edge-selection performance, indicating that noise updating is not uniformly beneficial across different noise-sample regimes.
In contrast, $\mathcal{H}$-Bayes remains sensitive to the loss-scaling parameter $w$: small values of $w$ provide nearly perfect recovery, whereas increasing $w$ leads to more false-positive detections and, under the CI-based rule, markedly poorer coverage.
Overall, NC-Bayes performs reliably for $\nu=1$ and for larger noise samples without adaptive noise updating, while $\mathcal{H}$-Bayes requires careful tuning of $w$ to maintain sparse and well-calibrated inference.

\begin{table}[b!]
\centering
\caption{
Performance metrics for Erd\"os--R\'enyi random graphs ($d=30$) under the median-based edge detection rule,
averaged over 100 simulations.
An edge is detected when the absolute posterior median of any associated coefficient exceeds 0.100.
}
\vspace{2mm}
\begin{tabular}{llllrrr}
\hline
Method & $\nu$ & $w$ & Noise Update & Recall & Precision & Accuracy \\
\hline\hline
\multirow{4}{*}{NC-Bayes}
& \multirow{2}{*}{1} & -- & False & 0.973 & 1.000 & 0.997 \\
&                        & -- & True  & 0.987 & 0.991 & 0.998 \\
\cline{2-7}
& \multirow{2}{*}{5} & -- & False & 1.000 & 0.998 & 1.000 \\
&                        & -- & True  & 0.741 & 0.110 & 0.409 \\
\hline
\multirow{3}{*}{$\mathcal{H}$-Bayes}
& -- & 0.2 & -- & 0.998 & 1.000 & 1.000 \\
& -- & 0.5 & -- & 1.000 & 0.973 & 0.997 \\
& -- & 1.0 & -- & 1.000 & 0.669 & 0.953 \\
\hline
\end{tabular}
\label{tab:results_ER_Median}
\end{table}

\begin{table}[b!]
\centering
\caption{
Performance metrics for Erd\"os--R\'enyi random graphs ($d=30$) under the 90\% CI-based edge detection rule,
averaged over 100 simulations.
An edge is detected when its 90\% CI does not contain zero.
CP denotes the coverage probability for $\phi$.
}
\vspace{2mm}
\begin{tabular}{llllrrrr}
\hline
Method & $\nu$ & $w$ & Noise Update & CP (\%) & Recall & Precision & Accuracy \\
\hline\hline
\multirow{4}{*}{NC-Bayes}
& \multirow{2}{*}{1} & -- & False & 98.9 & 0.776 & 1.000 & 0.979 \\
&                        & -- & True  & 98.4 & 0.841 & 1.000 & 0.985 \\
\cline{2-8}
& \multirow{2}{*}{5} & -- & False & 96.7 & 0.997 & 1.000 & 1.000 \\
&                        & -- & True  & 99.6 & 0.110 & 0.998 & 0.916 \\
\hline
\multirow{3}{*}{$\mathcal{H}$-Bayes}
& -- & 0.2 & -- & 92.9 & 0.988 & 1.000 & 0.999 \\
& -- & 0.5 & -- & 84.6 & 1.000 & 0.947 & 0.995 \\
& -- & 1.0 & -- & 67.2 & 1.000 & 0.266 & 0.740 \\
\hline
\end{tabular}
\label{tab:results_ER_CI}
\end{table}

\end{document}